\documentclass[final,3p,times]{elsarticle}

\usepackage{amsmath}
\usepackage{subfigure}
\usepackage{graphicx}
\usepackage{algorithmic}
\usepackage{algorithm}
\usepackage{xcolor}
\usepackage{amsfonts}
\usepackage{multirow}
\usepackage[normalem]{ulem}

\newtheorem{definition}{Definition}[section]
\newtheorem{theorem}{Theorem}[section]

\newtheorem{proof}{Proof}[section]
\newtheorem{lemma}{Lemma}[section]

\newtheorem{example}{Example}[section]
\newtheorem{observation}{Observation}[section]
\newtheorem{stipulation}{Stipulation}[section]

\journal{Information Sciences}

\begin{document}

\begin{frontmatter}



\title{Efficient Semi-External Depth-First Search}

\author[hit]{Xiaolong Wan}
\ead{wxl@hit.edu.cn}
\author[hit]{Hongzhi Wang\corref{cor1}}
\ead{wangzh@hit.edu.cn}

\cortext[cor1]{Corresponding author}
\address[hit]{School of Computer Science and Technology, Harbin Institute of Technology, China}

\begin{abstract}
	As the sizes of graphs grow rapidly, currently many real-world graphs can hardly be loaded in the main memory. 
	It becomes a hot topic to compute depth-first search~(DFS) results, i.e., depth-first order or DFS-Tree, on semi-external memory model. Semi-external algorithms assume the main memory could at least hold a spanning tree $T$ of a graph $G$, and gradually restructure $T$ into a DFS-Tree, which is non-trivial. In this paper, we present a comprehensive study of semi-external DFS problem. Based on our theoretical analysis of its main challenge, we introduce a new semi-external DFS algorithm, named EP-DFS, with a lightweight index $\mathcal{N}^+\negthickspace$-\textit{index}. Unlike traditional algorithms, we focus on addressing such complex problem efficiently not only with less I/Os, but also with simpler CPU calculation (implementation-friendly) and less random I/O accesses~(key-to-efficiency). Extensive experimental evaluation is conducted on both synthetic and real graphs. The experimental results confirm that our EP-DFS algorithm significantly outperforms existing algorithms.
\end{abstract}


\end{frontmatter}

\section{Introduction}\label{sec:introduction}

Depth-first Search~(DFS) is a basic way to learn graph properties from node to node, which is widely utilized in the graph field~\cite{DBLP:books/daglib/0037819}. To visit a node $u$ in a graph $G$, DFS first marks $u$ as \textit{visited}. Then, it recursively visits all the adjacent nodes of $u$ that are unmarked. Specifically, if DFS starts from a node that connects with all the others, the total order that DFS visits the nodes of $G$ is called depth-first order. Besides, the path that DFS walks on is a spanning tree~\cite{DBLP:books/daglib/0037819}, known as DFS-Tree. Below, we present an example about depth-first order and DFS-Tree.

\begin{example}
	\label{example:1.1}
	We draw a graph $G_0$ by three different ways shown in Figure~\ref{fig:introductionFigure}, where $T_a$, $T_b$ and $T_c$ are the spanning trees of $G_0$ constituted by the solid lines. In each subfigure of Figure~\ref{fig:introductionFigure}, the numbers around the nodes are the depth-first orders of the nodes on the spanning tree~(not $G_0$) related to the subfigure. According to the above discussion of DFS algorithm, $T_a$ is not a DFS-Tree of $G_0$, because after visiting node $p$, $T_a$ visits node $b$ instead of node $f$. $T_b$ and $T_c$ are both DFS-Trees of $G_0$, and the depth-first orders of the nodes in $G_0$ could be either $r,a,d,p,f,b,g,q,c,h$ or $r,b,g,q,f,c,h,a,d,p$.
\end{example}

Computing depth-first order or DFS-Tree is a key operation for many graph problems, such as finding strongly connected components~\cite{DBLP:conf/sigmod/ZhangYQCL13}, topology sort~\cite{10.5555/500824}, reachability query~\cite{DBLP:conf/sigmod/JinHWRX10}, etc. That makes DFS a fundamental operation in the graph field. For example, current algorithms for finding strongly connected components~(SCCs) have to find the DFS-Tree $T$ of $G$ or compute its total depth-first order. For example, Kosaraju-Sharir algorithm~\cite{DBLP:books/daglib/0037819} executes DFS to obtain a total depth-first order of $G$, and executes DFS again on the transposed $G$ for finding all the SCCs of $G$. Furthermore, topology sort occurs from a common problem with $n$ variables, known that some of them are less than some others. One has to check whether the given constraints are contradictory, and if not, an ascending order of these $n$ variables is required. Current solutions have to first obtain a DFS-Tree $T$ of $G$. One reason is when $G$ has a cycle, it has no topology order, and to check that, normally a SCC algorithm is used. Another is when $G$ has no cycle, a topology order of $G$ is the reversed postorder of the nodes on $T$~\cite{DBLP:books/daglib/0037819}.

Given a graph $G$ with $n$ nodes and $m$ edges, an in-memory DFS algorithm requires $O(n+m)$~\cite{DBLP:books/daglib/0037819} time. However, as the sizes of graphs grow rapidly in practical applications, loading many large-scale graphs into the main memory can hardly be done~\cite{LKAQ}. For example, at the end of 2014, Freebase~\cite{Freebase} included 68 million entities, 1 billion pieces of relationships and more than 2.4 billion factual triples. Currently, except in-memory algorithms, there are also two kinds of DFS algorithms: \textit{external algorithms} and \textit{semi-external algorithms}.

\textit{External algorithms.} Supposing one block of disk contains $B$ elements. External DFS algorithms assume that memory could hold at most $M$ elements, while the others would be stored on disk. Each I/O operation reads a block. However, as DFS may access the nodes of $G$ randomly, the fastest external DFS algorithm~\cite{DBLP:conf/soda/BuchsbaumGVW00, DBLP:conf/sigmod/ZhangYQS15} still needs $O((n+\frac{m}{B})log{\frac{n}{B}}+ sort(m))$ I/Os, where $sort(m)= O(\frac{m}{B}\log_{\frac{M}{B}}{\frac{n}{B}})$. Thus, when $G$ is relatively large, these external algorithms can hardly be utilized in practice in terms of their high time and I/O consumption.

\begin{figure}[t]
	\centering
\renewcommand{\thesubfigure}{}
\subfigure[(a) $T_a$]{
	\includegraphics[scale = 0.4]{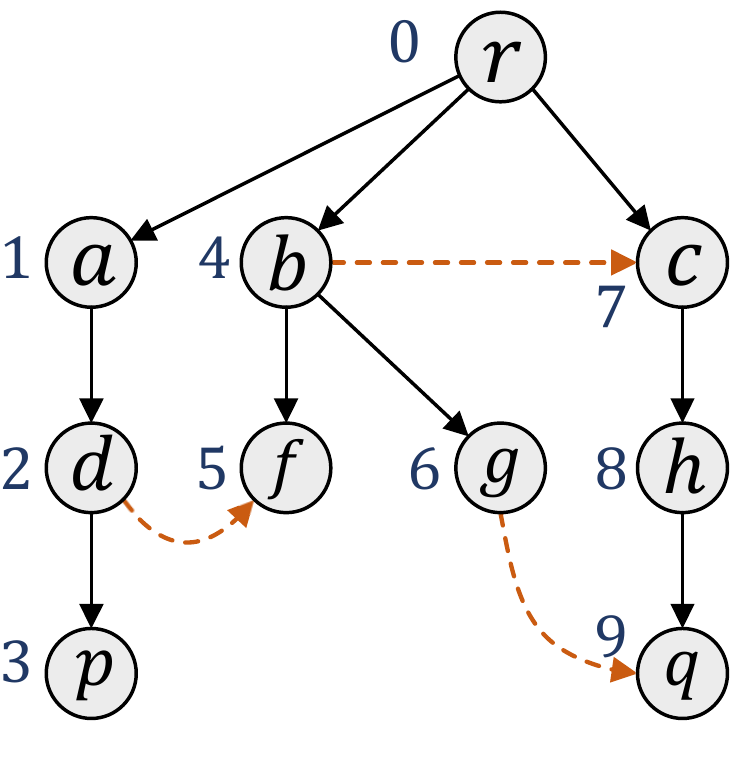}}
\quad
\subfigure[(b) $T_b$]{
	\includegraphics[scale = 0.4]{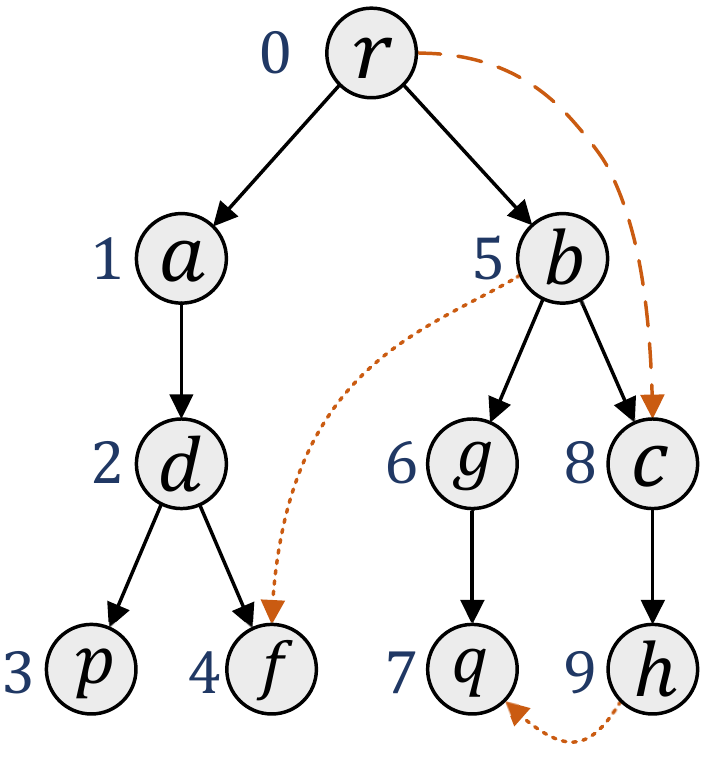}}
\quad
\subfigure[(c) $T_c$]{
	\includegraphics[scale = 0.4]{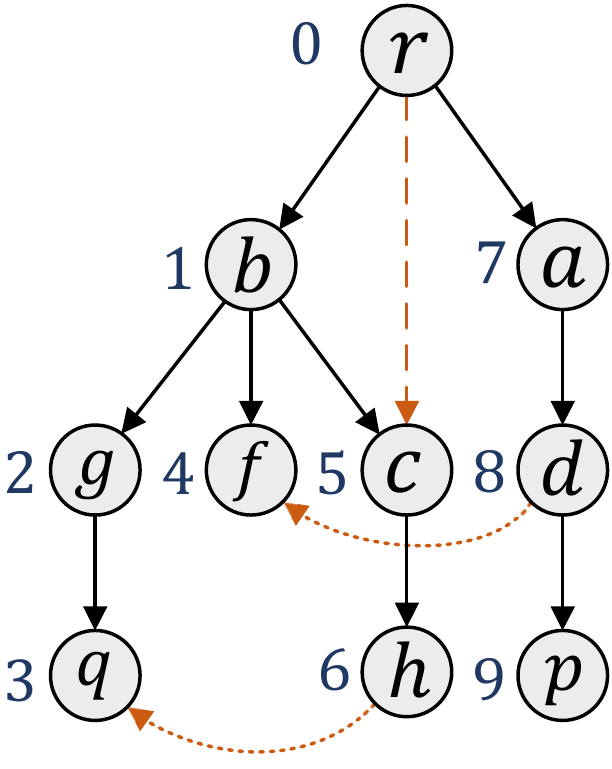}}

\caption{Schematic views of the spanning trees of a given graph $G_0$.}
\label{fig:introductionFigure}
\end{figure} 

\textit{Semi-external algorithms.} They assume that $c\times n$ elements could be maintained in the main memory, where $c$ is very small constant, e.g. $c=3$. Besides, instead of loading the whole graph, they construct an in-memory spanning tree $T$ of $G$, and convert the DFS problem into the problem of finding a DFS-Tree by restructuring $T$ with iterations. Computing DFS results on semi-external memory model is important, since the large-scale graphs, e.g. Freebase, can hardly be processed in the external memory.

Traditional technique for constructing DFS-Tree of a graph $G$ is introduced by Sibeyn et al.~\cite{DBLP:conf/spaa/SibeynAM02}, named \textit{EdgeByBatch}~(EB-DFS). It gradually adjusts $T$ into a DFS-Tree by iteratively scanning $G$ sequentially. Each iteration is called a round. EB-DFS ends in $k$th round when $T$ does not change from $(k-1)$th round to $k$th round, i.e. a DFS-Tree is obtained. Regardless the other I/O or computing cost, in practice, that is too expensive and meaningless to scan the entire $G$ for an extra round~($k$th round), as the size of $G$ is often relatively large. To address the high I/O cost of EB-DFS, Zhang et al.~\cite{DBLP:conf/sigmod/ZhangYQS15} devise a divide-and-conquer algorithm, named \textit{DivideConquerDFS}~(DC-DFS). After first executing one round of the above EB-DFS algorithm to construct a spanning tree $T$ of $G$, they divide $G$ into several small subgraphs, and recursively process such subgraphs. Unfortunately, in the graph division process, since the main memory cannot hold the entire graph, the division results need to be recorded in the external memory. DC-DFS is inevitably related to numerous random I/O accesses, especially when the number of the divided subgraphs is huge~\cite{ComputerSystems3}.

\textbf{Contributions.} In this paper, we provide a comprehensive study of the DFS problem on semi-external environment, which is crucial according to the above discussions. 

Even though the DFS problem has been studied over decades~\cite{DBLP:books/daglib/0037819}, addressing it on semi-external memory model is quite new, and only a few works in literature have been proposed for that because of its difficulty. To demonstrate why DFS problem on semi-external memory model is non-trivial, this paper presents a detailed discussion for its main challenge, from the perspective of theoretical analysis. 

Our analysis shows that the inefficiency of traditional semi-external DFS algorithms comes from a chain reaction during the process of iteratively restructuring $T$ to a DFS-Tree of $G$, where $T$ is initialized as a spanning tree of $G$. The chain reaction refers to that the depth-first orders of the nodes on $T$ are changed without predictability in traditional algorithms. Due to the chain reaction, in traditional algorithms, telling whether these edges belong to $T$ or not when $T$ is a DFS-Tree of $G$ is hard. Hence, to prune edges from $G$, many complex operations are devised in traditional algorithms to ensure their efficiencies, which may involve numerous random I/Os, or scanning $G$ many times meaninglessly, etc.

Motivated by that, a novel semi-external DFS algorithm is proposed, named EP-DFS. Firstly, EP-DFS can efficiently and effectively discard certain edges from $G$, even though the chain reaction exists in the process of iteratively restructuring $T$. The reason is that, in each iteration of EP-DFS, only a batch of edges $B$ is loaded into the main memory, where (\romannumeral1) $T$ is restructured into the DFS-Tree of the graph composed by $T$ and $B$. Here, $B$ contains a subset of edges in $G$, and each edge in $B$ is related to a node whose depth-first order on $T$ is in a certain range. It can ensure a large number of nodes in $G$ have fixed depth-first orders on $T$ in the early stage of the restructuring process. Hence, the edges related to such nodes could be discarded effectively. Secondly, to avoid numerous random disk accesses when loading $B$, a lightweight index, named $\mathcal{N}^+\negmedspace$-\textit{index}, is devised in EP-DFS. The key of this index is selecting the right set of edges from $G$, since EP-DFS only loads $B$ from $\mathcal{N}^+\negmedspace$-\textit{index}. The smaller the number of edges contained in $\mathcal{N}^+\negmedspace$-\textit{index}, the faster $B$ can be obtained from $\mathcal{N}^+\negmedspace$-\textit{index}, and the less I/O EP-DFS consumes.  To evaluate the performance of EP-DFS, we conduct extensive experiments on both real and synthetic datasets. Experimental results show that EP-DFS is efficient, and significantly outperforms traditional algorithms on various conditions.

Our main contributions are as follows:

\begin{itemize}	
	\item[-] This paper presents a detailed discussion about why the semi-external DFS problem is non-trivial.
	
	\item[-] This paper proposes a novel semi-external DFS algorithm, named EP-DFS, in which a large number of nodes on $T$ have fixed depth-first orders so that EP-DFS can efficiently prune edges from $G$.
	
	\item[-] This paper devises a lightweight index, named $\mathcal{N}^+\negmedspace$-\textit{index}, with which EP-DFS could access the edges in $G$ fast.
	
	\item[-] Extensive experiments conducted on both real and synthetic datasets confirm that the performance of our EP-DFS considerably surpasses that of the state-of-the-art algorithms. 
	
\end{itemize}

The rest of this paper is organized as follows. We introduce the preliminaries about the problem of semi-external DFS in  Section~\ref{sec:preliminaries}. Section~\ref{sec:related_works} gives a brief review of the existing solutions. The discussion about the main challenge of semi-external DFS problem is presented in Section~\ref{sec:chain_reaction}, to conquer which a naive algorithm is proposed in Section~\ref{sec:overview}. After that, we illustrate our EP-DFS algorithm in Section~\ref{sec:algorithm}. The experimental results are demonstrated in Section~\ref{sec:experiments}. We draw a conclusion in Section~\ref{sec:conclusion}.

\section{Preliminaries} \label{sec:preliminaries}
We study the problem of semi-external DFS on directed disk-resident graphs. Figure~\ref{fig:edge_order_example_T1} depicts an example of this section; Table~\ref{tab:notation} summarizes the frequently-used notations of this paper.

\begin{figure}[h]
	\centerline{\includegraphics[scale = 0.4]{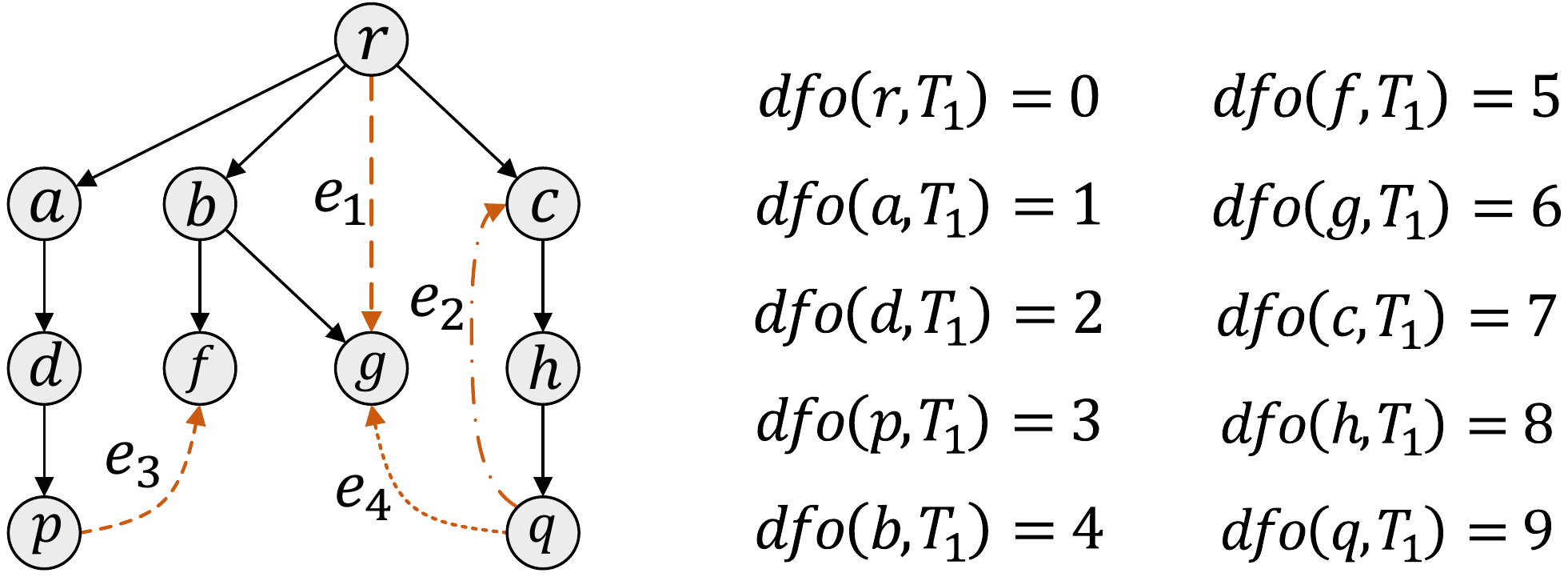}}
	\caption{A schematic view of a spanning tree $T_1$ of a graph $G_1$, where the black solid lines are the tree edges of $T_1$.}
	\label{fig:edge_order_example_T1}
\end{figure}

\begin{definition}\label{def:directedGraph}
	A directed graph $G$ is a tuple $(V,E)$, where 
	
	(\romannumeral1) $V$ and $E$ denote the node set and edge set of $G$, respectively, and $n=|V|$ and $m=|E|$; 
	
	(\romannumeral2) $E\subseteq V\times V$, each entry $e$ in $E$ is a directed edge, denoted by $(u,v)$, and $\forall e=(u,v)\in E$, $u$ is the tail of $e$, and $v$ is the head of $e$.
\end{definition}

For simplicity, we let $V(G)$ denote the node set of $G$, and let $E(G)$ denote the edge set of $G$. For instance, $G_1$ of Figure~\ref{fig:edge_order_example_T1} is a directed graph with $10$ nodes~($|V(G_1)|=10$), and $12$ directed edges~($|E(G_1)|=12$). Note that, we assume $G$ does not contain self-loops and multiple edges, since these edges are irrelevant to the correctness of the constructed DFS-Tree. Thus, each node has less than $n$ out-neighbors.

\begin{table}[t]
	\centering \caption{Frequently-used notations.} \label{tab:notation}
	\small
	\begin{tabular}{c|c|p{8cm}}
		\hline
		\textbf{Notation} & \textbf{Definition} & \textbf{Description}\\%
		\hline
		$G$ & \ref{def:directedGraph}		&	The input graph with $n$ nodes and $m$ edges.\\
		
		$d\negthinspace f\negthinspace o(v,T)$	&\ref{def:lambda}&	The depth-first order of $v$ on tree $T$.\\
		
		$T$	& \ref{def:spanningTree}&The in-memory spanning tree of $G$ rooted at node $r$.\\
		
		$\mathcal{C}(T,T')$  &\ref{def:S(T,T_i)}& A concrete number related to two spanning trees $T$ and $T'$ of $G$.\\
		
		$\Upsilon(T)$ &\ref{def:U(T)}& A concrete number related to $T$.\\
		
		$FNN$ &\ref{def:FNN}& A parameter utilized in EP-DFS.\\
		
		$\mathcal{B}$, $\mathcal{B}^+$&\ref{def:B}, \ref{def:B^+}& The edge batches of $G$, which are subsets of $E$.\\
		
		\hline
	\end{tabular}
\end{table}

\textit{Depth-first search~(DFS).} For a given graph $G$, DFS first visits a node $u$ and \textit{declares} $u$ as \textit{visited}. Then, it picks up an undeclared node $v$, from the out-neighborhood of the most recently visited node $w$. We say DFS \textit{walks} on such edges $(w,v)$. In general, we assume that $G$ has a node $r$ which connects to all the others. If DFS visits $G$ from $r$, then it could visit each node in $G$ once. The total order that DFS declares the nodes as \textit{visited} is known as \textbf{depth-first order}. This order is not unique, as demonstrated in Example~\ref{example:1.1}. For a depth-first order of $G$, there is a corresponding \textbf{DFS-Tree} $T$, composed by the edges that are passed by DFS. The examples of the DFS results, i.e. depth-first order and DFS-Tree, are demonstrated in Example~\ref{example:1.1}.

In the rest of this paper, notation $d\negthinspace f\negthinspace o(u,T)$ is utilized, as defined below.

\begin{definition}
	\label{def:lambda}
	$d\negthinspace f\negthinspace o(v,T)=k$ represents the depth-first order of $v$ on $T$ is $k$, where $T$ is a tree, and $v$ is a node of $T$.
\end{definition}

As an example, for each node $v$ in the graph $G_1$ depicted in Figure~\ref{fig:edge_order_example_T1}, the value of $d\negthinspace f\negthinspace o(v,T_1)$ is presented on the right side of  Figure~\ref{fig:edge_order_example_T1}.

In order to make sure that a DFS-Tree of $G$ only corresponds to a total depth-first order of $G$, we say that the spanning trees $T$ of $G$ are \textit{ordered} trees. 

\begin{definition}
	\label{def:spanningTree}
	A spanning tree $T$ of $G$ is an ordered tree, where the out-neighbors of each node $u$ on $T$ is arranged from left to right, by the following rule: $\forall u, v\in V(T)$, $u$ is the \textbf{left brother} of $v$, or $v$ is the \textbf{right brother} of $u$, iff, (\romannumeral1) $u$ and $v$ have the same parent node on $T$ and (\romannumeral2) $d\negthinspace f\negthinspace o(u,T)<d\negthinspace f\negthinspace o(v,T)$.
\end{definition}

That is, there is an order among all the children $u_1,u_2,\dots,u_k$ of each non-leaf node $u$ of $T$. If $d\negthinspace f\negthinspace o(u_1,T)<d\negthinspace f\negthinspace o(u_2,T)<\dotsb<d\negthinspace f\negthinspace o(u_k,T)$, then $u_1$ is the leftmost child of node $u$ in $T$ while $u_k$ is the rightmost child of node $u$ in $T$. The running examples, in Figures~\ref{fig:introductionFigure}-\ref{fig:rearrangement_example}, are drawn, according to the above order.

In general, when a semi-external algorithm uses an in-memory DFS algorithm to obtain a DFS-Tree of a graph $G_{T\cup\mathcal{E}}$ composed by $T$ and $\mathcal{E}$~($\mathcal{E}\subseteq E(G)$), it should follow Stipulation~\ref{sti:stipulation}. 
\begin{stipulation}
	\label{sti:stipulation}
	For each node $u$ in $G_{T\cup\mathcal{E}}$, after visiting $u$, (\romannumeral1) if a node in the out-neighborhood of $u$ on $T$ has not been visited, then DFS first visits the leftmost child of $u$ on $T$ which has not been visited; (\romannumeral2) otherwise, DFS picks a node $v$ from $\mathcal{E}$, where there is an edge $(u,v)$ in $\mathcal{E}$, and $v$ has not been visited.
\end{stipulation}

For instance, as shown in Figure~\ref{fig:introductionFigure}, $G_0$ is a graph composed by $T_a$ and an edge list $\mathcal{E}=\{(d,f),(g,q),(b,c)\}$. $T_b$ is the DFS-Tree of $G_0$ obtained under Stipulation~\ref{sti:stipulation}, while $T_c$ is not. Note that, in Figure~\ref{fig:introductionFigure}, the numbers related to the nodes are their depth-first orders on $T_a$, $T_b$ or $T_c$.

Furthermore, given a spanning tree $T$ of $G$, the edges in $G$ could be classified, as illustrated in Definition~\ref{def:edge_types}. 

\begin{definition}
	\label{def:edge_types}
	If an edge $e$ in $T$, $e$ is a tree edge, otherwise, $e$ is a non-tree edge.  For a non-tree edge $e=(u,v)$,  
	
	(\romannumeral1) $e$ is a forward edge, iff, the LCA~(Least Common Ancestor) of $u$ and $v$ on $T$ is $u$ and $(u,v)\notin E(T)$; 
	
	(\romannumeral2) $e$ is a backward edge, iff, the LCA of $u$ and $v$ on $T$ is $v$;
	
	(\romannumeral3) $e$ is a forward cross edge, iff, $d\negthinspace f\negthinspace o(u,T)<d\negthinspace f\negthinspace o(v,T)$ and the LCA of $u$ and $v$ on $T$ is not $u$; 
	
	(\romannumeral4) $e$ is a backward cross edge, iff, $d\negthinspace f\negthinspace o(u,T)>d\negthinspace f\negthinspace o(v,T)$ and the LCA of $u$ and $v$ on $T$ is not $v$.
\end{definition}
For instance, in Figure~\ref{fig:edge_order_example_T1}, $(a,d)$ is a tree edge, $e_1$ is a forward edge, $e_2$ is a backward edge, $e_3$ is a forward cross edge, and $e_4$ is a backward cross edge, as classified by $T_1$ of $G_1$.

\textbf{Problem statement.} We study the semi-external DFS problem with the restriction that \textit{at most $2n$ edges could be hold in the main memory}. With such restriction, the semi-external DFS problem is much more complicated and valuable, in that: (\romannumeral1)  the greater the number of edges that can be maintained in memory, the better the performance of a semi-external DFS algorithm; (\romannumeral2) the efficient state-of-the-art algorithms~(EB-DFS and DC-DFS), require to maintain at least $2n$ edges in the main memory, which is discussed in Section~\ref{sec:related_works}.

\section{Existing solutions}\label{sec:related_works}
The problem of semi-external DFS is originated in \cite{DBLP:conf/spaa/SibeynAM02}, which assumes that only $c\times n$ elements could be loaded into memory~($c$ is a small constant). That assumption is important. For one thing, with the increases of the sizes of current graph databases, loading a large-scale graph into the main memory becomes harder and harder. For another, to process a large-scale graph, external DFS algorithms are extremely inefficient, as discussed in Section~\ref{sec:introduction}. According to the fact that ``\textit{Given a spanning Tree $T$ of $G$, $T$ is a DFS-Tree of $G$, iff, $G$ has no forward cross edge as classified by $T$}''~\cite{DBLP:conf/spaa/SibeynAM02,DBLP:conf/sigmod/ZhangYQS15}, three main algorithms are proposed, i.e. \textit{EE-DFS}~(EdgeByEdge), \textit{EB-DFS}~(EdgeByBatch) and \textit{DC-DFS}~(DivideConquerDFS). 

\textit{EE-DFS.} It is proposed in \cite{DBLP:conf/spaa/SibeynAM02}, which computes the DFS-Tree of $G$ by iterations. In each iteration, it sequentially scans $G$ a time. Assuming $w$ is the parent of $v$ on $T$. In each iteration, if an edge $e=(u,v)$ is a forward cross edge as classified by $T$, then EE-DFS lets $v$ be the rightmost child $u$ on $T$, and removes edge $(w,v)$ from $T$. EE-DFS ends in $k$th iteration, if, $\forall e\in E$, $e$ is not a forward cross edge as classified by $T$, during the $k$th scanning process. EE-DFS is inefficient, whose major drawback is that, for each scanned edge $e=(u,v)$, it needs to compute the LCA of $u$ and $v$ on $T$, in order to determine whether $e$ is a forward cross edge as classified by $T$. As $T$ is changed dynamically~(no preprocess is allowed), the time complexity for answering LCA queries of each edge in $E$ is too high to be afforded\footnote{https://algotree.org/algorithms/lowest\_common\_ancestor/}~\cite{DBLP:conf/spaa/SibeynAM02}.

\textit{EB-DFS.} Different from EE-DFS, EB-DFS~\cite{DBLP:conf/spaa/SibeynAM02} processes the edges in $G$ by batch to avoid the operation of computing the edge types. It calls a sequential scan of $G$ a round. In each round, for each $n$ edges $E_{n}$, EB-DFS replaces $T$ with the DFS-Tree of graph $G'$, where $G'$ is composed of $T$ and $E_{n}$. In addition, a function \textit{Reduction-Rearrangement} is utilized in EB-DFS, which requires to scan the entire input graph one more time, in order to classify the nodes on $T$ into two kinds: \textit{passive} and \textit{non-passive}. The out-going edges of the passive nodes could be passed in the later round. For each non-passive node $v$ with $k$ children $u_1,\dots,u_k$, they rank the children of $v$ according to their weights. Here, the weight of a node $u_i$ in $T$ is the size of the biggest subtree of $T$ rooted at $u_i$.

EB-DFS requires maintaining $2n$ edges in memory by default, for ensuring its efficiency. Even though it is possible to implement EB-DFS by letting it only maintaining $1.5n$ edges or even fewer edges in the main memory, its performance decreases greatly. When it only loads $n+1$ edge in the main memory~(the smallest number), its performance will be worse than that of EE-DFS, as discussed in \cite{DBLP:conf/spaa/SibeynAM02}. 

In EB-DFS, each node in $T$ has three attributes~\cite{DBLP:conf/spaa/SibeynAM02} required by its procedure Reduction-Rearrangement. As far as we know, there is no technique proposed for packing all the node attributes of a semi-external algorithm. That is because, the values of the node attributes change constantly, and it is  difficult to estimate them in advance. To ensure efficiency, semi-external algorithms prefer to apply for a contiguous memory space to maintain their node attributes, instead of packing them. 

EB-DFS, in the worst case, needs $n$ times graph scanning, when $G$ is a strongly connected graph with $n$ nodes and $2n$ edges positioned as a cycle. To address that, two additional functions are developed, which, unfortunately, \textit{is quite expensive, and should be performed only when necessary},  according to \cite{DBLP:conf/spaa/SibeynAM02}.

\textit{DC-DFS.} That approach~\cite{DBLP:conf/sigmod/ZhangYQS15} is a divide-and-conquer algorithm. DC-DFS aims to iteratively divide $G$ into several small graphs equally and correctly, with two division algorithms \textit{Divide-Star} and \textit{Divide-TD} which both utilize the data structure \textit{S-Graph} $\Sigma$. The former constructs $\Sigma$ based on $T$, where (\romannumeral1) initially $\Sigma$ contains the root $r$ of $T$ and all the children of $r$ in $T$; (\romannumeral2) scanning all the edges of $G$ from disk, and, $\forall e =(u,v)\in E$, if the LCA of $u$ and $v$ on $T$ is $r$~($r\neq u$ and $r\neq v$), computing the \textit{S-edge} $e'$ of $e$ and adding $e'$ into $\Sigma$; (\romannumeral3) restructuring $\Sigma$ if $\Sigma$ is not a \textit{directed acyclic graph} by a \textit{node contraction operation} \cite{DBLP:conf/sigmod/ZhangYQS15}. Then the $G$ is divided based on $\Sigma$. The latter is similar to Divide-Star, except that it initializes $\Sigma$ by a \textit{cut-tree} $T_c$ of $T$, where (\romannumeral1) the complete graph composed by the nodes in $T_c$ can fit into the main memory; (\romannumeral2) if a node $u$ in $T_c$, then $u$ is also in $T$; (\romannumeral3) if the children of $u$ in $T_c$ are $v_1,v_2,\dots,v_k$, then the children of $u$ in $T$ are $v_1,v_2,\dots,v_k$. 

DC-DFS involves numerous random I/O accesses when the number of the divided subgraphs is huge, comparing to EB-DFS which only accesses $G$ sequentially.  The reason is as follows. $G$ is disk-resident e.g. in a file $F_G$. The division process of $G$, in other words, could be treated as sequentially splitting and restoring $F_G$ into several small files, where the elements in each small file are not necessarily continuous in $F_G$. In addition, in such case, the scales of the most divided subgraphs are relatively small, as a DFS-Tree tends to a left deep tree according to Stipulation~\ref{sti:stipulation}. Hence, the elements of a large-scale divided subgraph are hard to be stored continuously on disk, according to~\cite{ComputerSystems3}.

DC-DFS needs at least to maintain $2n$ edges of $G$ in memory for the in-memory spanning tree $T$ and the S-Graph $\Sigma$, when the recursion depth of DC-DFS is relatively large. In addition, it requires at least $5n$ space to support its division process, as discussed below. DC-DFS highly relies on certain in-memory graph algorithms. Especially, it requires to parse all the edge types, in each division process. Hence, it is inevitable to compute the LCA of $u$ and $v$ on $T$ for each edge $(u,v)$ in $G$ at least once. With the consideration that the structure of $T$ is static during the parsing process of the edge types, and $m$ is often huge, \textit{Farach-Colton and Bender} algorithm~\cite{DBLP:journals/jal/BenderFPSS05} is the only option\footnote{Many thanks for the website ``https://cp-algorithms.com/''.}. Based on the data structure in \cite{RMQDataStructure}, \textit{Farach-Colton and Bender} requires at least $5n$ memory space.

\textbf{Remark.} The efficient traditional semi-external DFS algorithms, i.e. EB-DFS and DC-DFS, maintain at least $2n$ edges in the main memory. Besides, certain CPU calculations that they request are expensive, since (\romannumeral1) EB-DFS proposes two functions to address the problem caused by the worst case, but such functions are too expensive to be utilized in practical applications~(\textit{expensive CPU calculation}); (\romannumeral2) DC-DFS relies on a certain sophisticated algorithm, causing it requires $2n$ more additional memory space than EB-DFS~(\textit{expensive memory space consumption}). In addition, DC-DFS might be related to numerous random I/Os when the input graph is huge~(\textit{expensive disk access}).

Motivated by the limitations of the state-of-the-art approaches, we aim to address the DFS problem on semi-external environment with simpler CPU calculation, lower memory space consumption and fewer random disk accesses in practice. To achieve that goal, we put a lot of efforts to find out why the semi-external DFS problem is difficulty, as discussed in Section~\ref{sec:chain_reaction}.

%Our approach is based on the following analysis of  Section~\ref{sec:overview}, which contains, the theoretical analysis of the main challenge that causes the existing algorithms to be limited in certain cases.

\section{Problem Analysis: the main challenge ``chain reaction''}\label{sec:chain_reaction}

A semi-external DFS algorithm can easily find a total depth-first order or a DFS-Tree for $G$, when $G$ has less than $n$ edges. That is because, semi-external algorithms assume that the main memory at least has the ability to maintain a spanning tree $T$ of $G$. When $G$ has less than $n$ edges, all its edges could be loaded into the main memory. However, for a real-word graph $G$, $m$ is normally dozens of times or even hundreds of times larger than $n$. When $\frac{m}{n}$ is quite large, it is hard to compute the DFS results on semi-external memory model. Because, in this case, a semi-external algorithm can only maintain a portion of edges of $G$ in the main memory simultaneously. Some edges that are currently resided in the main memory must be removed and replaced by the new scanned edges, as discussed in Section~\ref{sec:related_works}.  

The difficulty of the semi-external DFS problem is caused by the total depth-first order of the nodes on $T$ is changed without predictability, when $T$ is replaced to the DFS-Tree of a subgraph of $G$, even under Stipulation~\ref{sti:stipulation}. That is as $T$ is changed constantly, the depth-first orders of the nodes on $T$ are changed correspondingly, so that the edge types as classified by $T$ are changed constantly. 

To present an efficient semi-external DFS algorithm, in this section, we present an in-depth study for how the edge types change when $T$ is replaced. Specifically, we assume that $G$ is divided into a series of edge batches, which are $B_1, B_2,\dots, B_i,\dots, B_k$, in that $G$ is stored on disk as an edge list. $T_i$ represents the DFS-Tree of the graph composed by $T$ and $B_i$. Without loss of generality, we first assume that $B_i=\{e\}$, where (\romannumeral1) $e$ happens to be a forward cross edge of $G$ as classified by $T$, and (\romannumeral2) $e=(x,u)$ or $e=(x,v)$. Nine groups of examples are drawn to show what happens to the types of the other edges when we restructure $T$ to $T_i$, as depicted in Figure~\ref{fig:chain_reaction}. The solid lines of Figure~\ref{fig:chain_reaction}(a), Figure~\ref{fig:chain_reaction}(b), $\dots$, and Figure~\ref{fig:chain_reaction}(i) are the tree edges of $T$, while edge $(u,v)$ is an edge of $G$. Correspondingly, the solid lines in Figure~\ref{fig:chain_reaction}(A), Figure~\ref{fig:chain_reaction}(B), $\dots$, and Figure~\ref{fig:chain_reaction}(I) are the tree edges of $T_i$.

Besides, one should know that the positional relationship between two nodes $u$ and $x$ on $T$ has only two possibilities, where (\romannumeral1) assuming $u$, $v$ and $x$ are the nodes of $T$; (\romannumeral2) $(x,v)$ is a forward cross edge as classified by $T$; (\romannumeral3) $u$ is an ancestor of $v$ on $T$. One is, if $d\negthinspace f\negthinspace o(x,T)$ is larger than $d\negthinspace f\negthinspace o(u,T)$, then $u$ is also an ancestor of $x$ on $T$. Another is, if $d\negthinspace f\negthinspace o(x,T)$ is smaller than $d\negthinspace f\negthinspace o(u,T)$, then $u$ and $x$ belongs to two different subtrees of $T$. The above discussion is formally defined the following Lemma~\ref{lemma:chain_reaction_lemma}. 

\begin{lemma}\label{lemma:chain_reaction_lemma}
Given nodes $u,v,x$ in $T$, $u$ is an ancestor of $v$, and $(x,v)$ is a forward cross edge as classified by $T$. If $d\negthinspace f\negthinspace o(x,T)>d\negthinspace f\negthinspace o(u,T)$, then the  LCA of $u$ and $x$ on $T$ is $u$; if $d\negthinspace f\negthinspace o(x,T)<d\negthinspace f\negthinspace o(u,T)$, then the LCA of $u$ and $x$ on $T$ is neither $u$ nor $x$.	
\end{lemma}
\begin{proof}
$(x,v)$ is a forward cross edge as classified by $T$. That is $d\negthinspace f\negthinspace o(x,T)<d\negthinspace f\negthinspace o(v,T)$ and the LCA of $x$ and $v$ on $T$ is neither $x$ nor $v$. On the one hand, if $d\negthinspace f\negthinspace o(x,T)>d\negthinspace f\negthinspace o(u,T)$, assuming $x$ is not a descendant of $u$, i.e. the LCA of $u$ and $x$ on $T$ is not $u$. As $d\negthinspace f\negthinspace o(x,T)>d\negthinspace f\negthinspace o(u,T)$, $x$ can only be a right brother of $u$, a descendant of $u$'s right brothers or a descendant of a $u$'s ancestor's right brother. Thus, $d\negthinspace f\negthinspace o(x,T)>d\negthinspace f\negthinspace o(v,T)$, which contradicts the premise. One the other hand, if $d\negthinspace f\negthinspace o(x,T)<d\negthinspace f\negthinspace o(u,T)$, assuming the LCA of $u$ and $x$ on $T$ is $x$. Then, the LCA of $x$ and $v$ on $T$ is $x$. Hence, the LCA of $u$ and $x$ on $T$ is neither $u$ nor $x$, as, when $d\negthinspace f\negthinspace o(x,T)<d\negthinspace f\negthinspace o(u,T)$, the LCA of $u$ and $x$ on $T$ is not $u$.
\end{proof}

\begin{figure}[t]
	\centerline{\includegraphics[scale = 0.85]{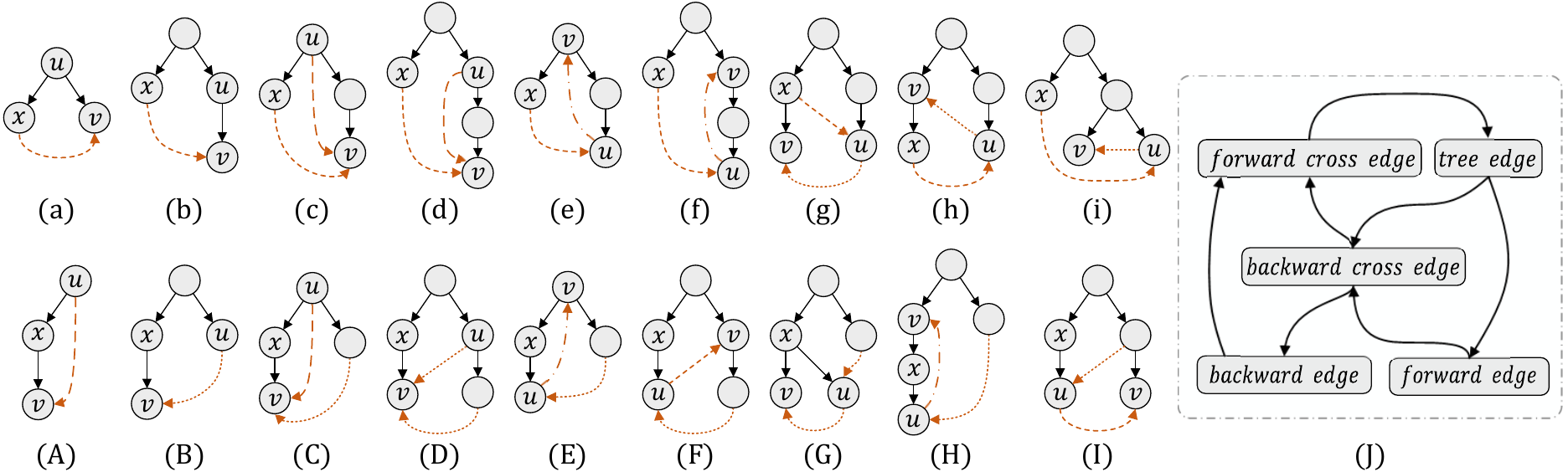}}
	\caption{The demonstration of the chain reaction of modifying the in-memory spanning tree $T$ with a forward edge $(x,v)$ or $(x,u)$. It is worth noting that all the self-loops are omitted in subfigure~(J).}
	\label{fig:chain_reaction}
\end{figure}

Firstly, we assume that $(u,v)$ is a tree edge of $T$. Then, when $B_i=\{e=(x,u)\}$, a semi-external algorithm restructures $T$ to $T_i$, by (1) letting $(x,u)$ be a tree edge of $T$; (2) removing $(w,x)$ from $T$, where node $w$ is the parent of $u$ on $T$. Thus, $(u,v)$ is still a tree edge of $T_i$. However, when $B_i=\{e=(x,v)\}$, $(u,v)$ is either a forward edge or a backward cross edge as classified by $T_i$. For one thing, when $d\negthinspace f\negthinspace o(x,T)$ is larger than $d\negthinspace f\negthinspace o(u,T)$, $(u,v)$ is a forward edge as classified by $T_i$. The reason is, according to Lemma~\ref{lemma:chain_reaction_lemma}, in this case, $u$ must be an ancestor of $x$ on $T$, as shown in Figure~\ref{fig:chain_reaction}(a). Thus, based on Stipulation~\ref{sti:stipulation}, $(x,v)$ should be a tree edge of $T$, and edge $(u,v)$ should be removed from $T$. That is $(u,v)$ is a forward edge as classified by $T_i$, as shown in Figure~\ref{fig:chain_reaction}(A). For another, when $d\negthinspace f\negthinspace o(x,T)$ is smaller than $d\negthinspace f\negthinspace o(u,T)$, $(u,v)$ is a backward cross edge as classified by $T_i$. That is because, according to Lemma~\ref{lemma:chain_reaction_lemma}, if $d\negthinspace f\negthinspace o(x,T)<d\negthinspace f\negthinspace o(u,T)$, node $u$ and node $x$ must belong to different subtrees of $T$, as shown in Figure~\ref{fig:chain_reaction}(b). Hence, a semi-external DFS algorithm should let $(x,v)$ be a tree edge of $T$, and remove $(u,v)$ from $T$, in terms of Stipulation~\ref{sti:stipulation}. That is $(u,v)$ is a backward cross edge as classified by $T_i$, as shown in Figure~\ref{fig:chain_reaction}(B).

After that, we assume $(u,v)$ is a forward edge as classified by $T$. A semi-external algorithm may let $(u,v)$ be a forward edge or a backward cross edge as classified by $T_i$. For one thing, if $B_i=\{e=(x,u)\}$, then after restructuring $T$ to $T_i$, $(u,v)$ is a forward edge as classified by $T_i$, which is obvious. In addition, when $B_i=\{e=(x,v)\}$ and $d\negthinspace f\negthinspace o(x,T)$ is larger than $d\negthinspace f\negthinspace o(v,T)$, then $(u,v)$ is also a forward edge as classified by $T_i$. Because according to Lemma~\ref{lemma:chain_reaction_lemma}, $u$ must be an ancestor of $x$ on $T$, as shown in Figure~\ref{fig:chain_reaction}(c) and Figure~\ref{fig:chain_reaction}(C). For another, if $B_i=\{e=(x,v)\}$ and $d\negthinspace f\negthinspace o(x,T)$ is smaller than $d\negthinspace f\negthinspace o(v,T)$, then $(u,v)$ is a backward cross edge as classified by $T_i$, in that according to Lemma~\ref{lemma:chain_reaction_lemma}, $u$ and $x$ must belong to different subtrees of $T$. Figure~\ref{fig:chain_reaction}(d) and Figure~\ref{fig:chain_reaction}(D) is a demonstration for this  case.

Moreover, if $(u,v)$ is a backward edge as classified by $T$, then $(u,v)$ is either a backward edge or forward cross edge as classified by $T_i$, which is discussed as follows. On the one hand, if $B_i=\{e=(x,v)\}$, $(u,v)$ must be a backward edge as classified by $T_i$, because restructuring $T$ to $T_i$ does not change the descendants of $v$ on the tree. On the other hand, when $B_i=\{e=(x,u)\}$, there are two cases. One is when $d\negthinspace f\negthinspace o(x,T)$ is larger than $d\negthinspace f\negthinspace o(v,T)$, according to Lemma~\ref{lemma:chain_reaction_lemma}, $v$ must be an ancestor of $x$ on $T$, as shown in Figure~\ref{fig:chain_reaction}(e). Thus, according to  Figure~\ref{fig:chain_reaction}(E), $(u,v)$ is a backward edge, as classified by $T_i$. Another is that if $d\negthinspace f\negthinspace o(x,T)$ is smaller than $d\negthinspace f\negthinspace o(v,T)$, then according to Lemma~\ref{lemma:chain_reaction_lemma}, $x$ and $v$ are in different subtrees of $T$, as shown in Figure~\ref{fig:chain_reaction}(f). Then, $(u,v)$ is a forward cross edge as classified by $T_i$, as demonstrated in Figure~\ref{fig:chain_reaction}(F).

Furthermore, edge $(u,v)$ is given as a backward cross edge. $(u,v)$ can be any type of the non-tree edges, except the type of forward edges, as classified by $T_i$. The reason is discussed below. When $(u,v)$ is a backward cross edge as classified by $T$, then $d\negthinspace f\negthinspace o(u,T)$ is smaller than $d\negthinspace f\negthinspace o(v,T)$, and nodes $u$ and $v$ belong to two different subtrees of $T$. On the one hand, assuming $B_i=\{e=(x,v)\}$, then $d\negthinspace f\negthinspace o(v,T_i)$ must be smaller than $d\negthinspace f\negthinspace o(v,T)$, and $d\negthinspace f\negthinspace o(u,T_i)$ is equal to $d\negthinspace f\negthinspace o(u,T)$. That is, $d\negthinspace f\negthinspace o(v,T_i)$ is smaller than $d\negthinspace f\negthinspace o(u,T_i)$. Thus, $(u,v)$ is not a forward edge as classified by $T_i$. On the other hand, when $B_i=\{e=(x,u)\}$, then $u$ cannot be an ancestor of $v$ on $T_i$, so that $(u,v)$ is not a forward cross edge as classified by $T_i$. Figure~\ref{fig:chain_reaction}(g)-Figure~\ref{fig:chain_reaction}(G), Figure~\ref{fig:chain_reaction}(h)-Figure~\ref{fig:chain_reaction}(H), and Figure~\ref{fig:chain_reaction}(i)-Figure~\ref{fig:chain_reaction}(I) are three groups of instances. Edge $(u,v)$ is a backward cross edge as classified by $T$ in Figure~\ref{fig:chain_reaction}(g), Figure~\ref{fig:chain_reaction}(h) and Figure~\ref{fig:chain_reaction}(i). Edge $(u,v)$ is a backward cross edge, a backward edge, and a forward cross edge as classified by $T_i$, as demonstrated in Figure~\ref{fig:chain_reaction}(G), Figure~\ref{fig:chain_reaction}(H) and Figure~\ref{fig:chain_reaction}(I), respectively.

Therefore, when $B_i$ has a forward cross edge as classified by $T$, there may be a set of edges in $G$ whose edge types as classified by $T$ have been changed, after $T$ is restructured to $T_i$, the DFS-Tree of the graph composed by $T$ and $B_i$. Since the distribution of the nodes and edges in $G$ is unknown, and the order of the edges stored on disk is also unknown, after restructuring $T$, it is impossible to predict which edges in $G$ will have a changed edge type, and what edge type will they become. Not to mention, it is very likely that $B_i$ contains more than one forward cross edge, when more than one edge is allowed to be loaded into memory. We notice that there is a ``chain reaction'' in the process of restructuring $T$, which refers to the cycles of Figure~\ref{fig:chain_reaction}(J). Here, Figure~\ref{fig:chain_reaction}(J) is summarized based on the above discussion, in which the self-loops are omitted. 

With the existence of the chain reaction, it is hard to determine whether an edge will be used in the future or not. Thus, traditional algorithms have to devise a series of complex operations to prune edges from $G$ or be related to numerous random I/O accesses~(discussed in Section~\ref{sec:related_works}). For example, in order to select the nodes whose out-neighborhoods should be rearranged, EB-DFS has to scan $G$ one more time at the end of each iteration. Besides, to improve the performance of EB-DFS when $G$ may contain a large complex cycle, two additional functions are devised for EB-DFS, which are quite expensive and can only be used when necessary. DC-DFS is presented for computing the DFS results to reduce the high time and I/O costs of EB-DFS. Unfortunately, since real-world graphs are complex, DC-DFS may have to divide the input graph $G$ into a large set of small subgraphs, so that it involves numerous random disk accesses.

\section{A naive algorithm}\label{sec:overview}

As the chain reaction exists, given a complex network $G$, a semi-external DFS algorithm may have to scan $G$ $n$ times, or be related to numerous random disk accesses, as discussed in Section~\ref{sec:related_works}. This cannot be afforded in practice. In order to reduce the effect of the chain reaction on the performance of restructuring $T$ to a DFS-Tree of $G$, an in-depth study is conducted, in which we have a simple but very important observation: 
\begin{itemize}
	\item[-] Given an edge batch $B_i=\{(u,v)\}$, assuming $T_i$ is the DFS-Tree of the graph composed by $T$ and $B_i$ under Stipulation~\ref{sti:stipulation}. Then, $\forall x\in V(T)$, $d\negthinspace f\negthinspace o(x,T)$ is equal to $d\negthinspace f\negthinspace o(x,T_i)$, if $d\negthinspace f\negthinspace o(x,T)$ is smaller than $d\negthinspace f\negthinspace o(v,T_i)$.   
\end{itemize}
In other words, the edge types of $e$ as classified by $T$ and $T_i$ are the same, when $e=(s,t)$ and the depth-first orders of $s$ and $t$ on $T$ are smaller than that of $v$ on $T_i$. For example, when $T$ is in the form of $\mathcal{T}_0$~(Figure~\ref{fig:basic_idea_example_G}(a)), given an edge batch $B_i=\{e_2=(b,c)\}$, a semi-external algorithm will restructure $T$ to $T_i$, where $T_i$ is in the form of $\mathcal{T}_1$~(Figure~\ref{fig:basic_idea_example_G}(b)). Then, the depth-first order of $c$ on $T_i$~($\mathcal{T}_1$) is $7$, and if a node has a depth-first order on $T$ that is smaller than $7$, then it has the same depth-first orders on $T$ and on $T_i$.

\begin{figure}[t]
	
	\centering
	\renewcommand{\thesubfigure}{}
	\subfigure[(a) $\mathcal{T}_0$]{
		\includegraphics[scale = 0.35]{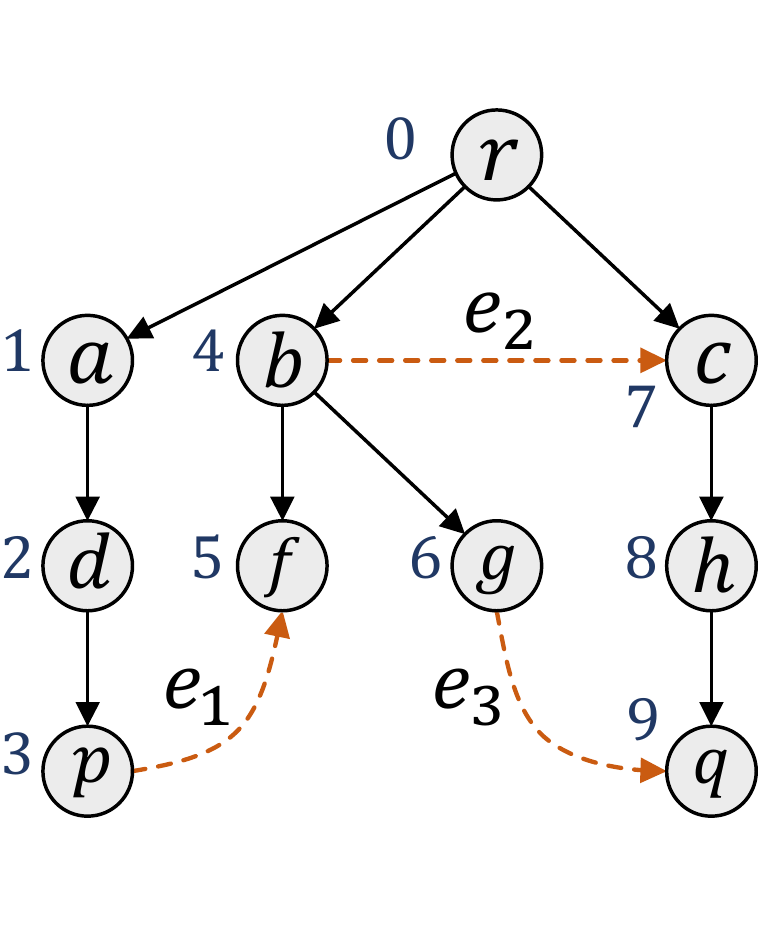}}
	\quad
	\subfigure[(b) $\mathcal{T}_1$]{
		\includegraphics[scale = 0.35]{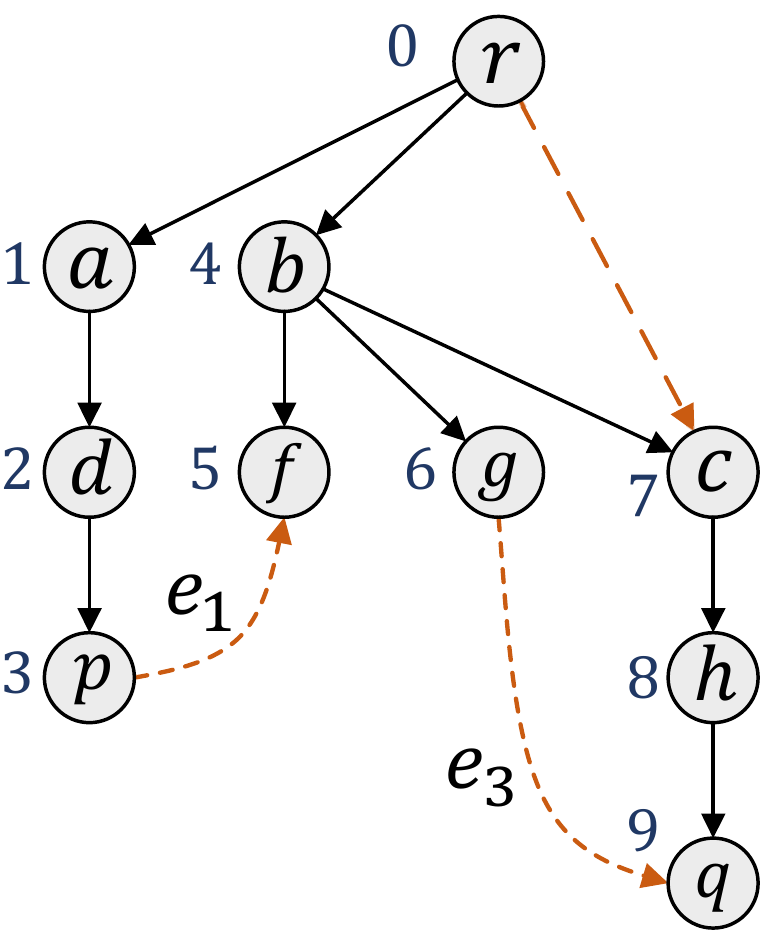}}
	\quad
	\subfigure[(c) $\mathcal{T}_2$]{
		\includegraphics[scale = 0.35]{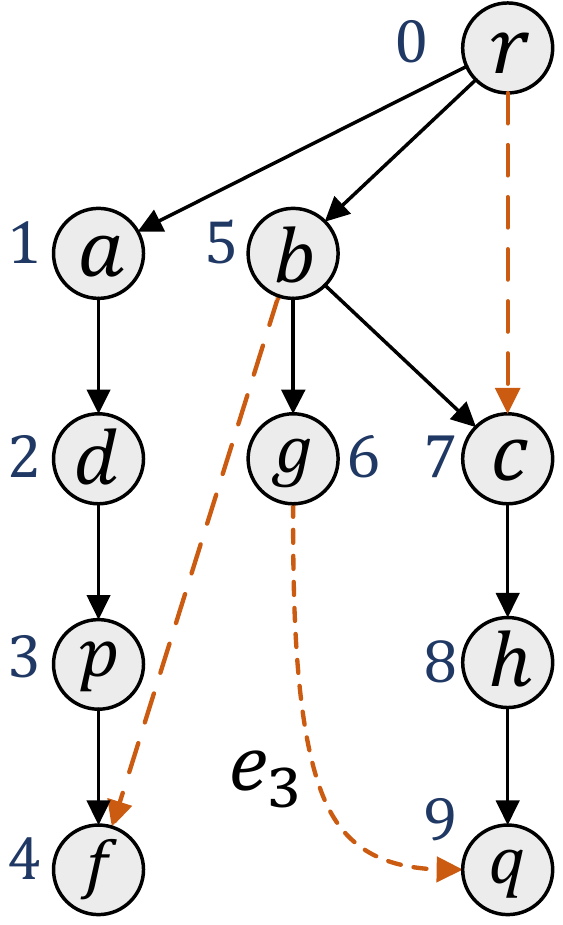}}
	\quad
	\subfigure[(d) $\mathcal{T}_3$]{
		\includegraphics[scale = 0.35]{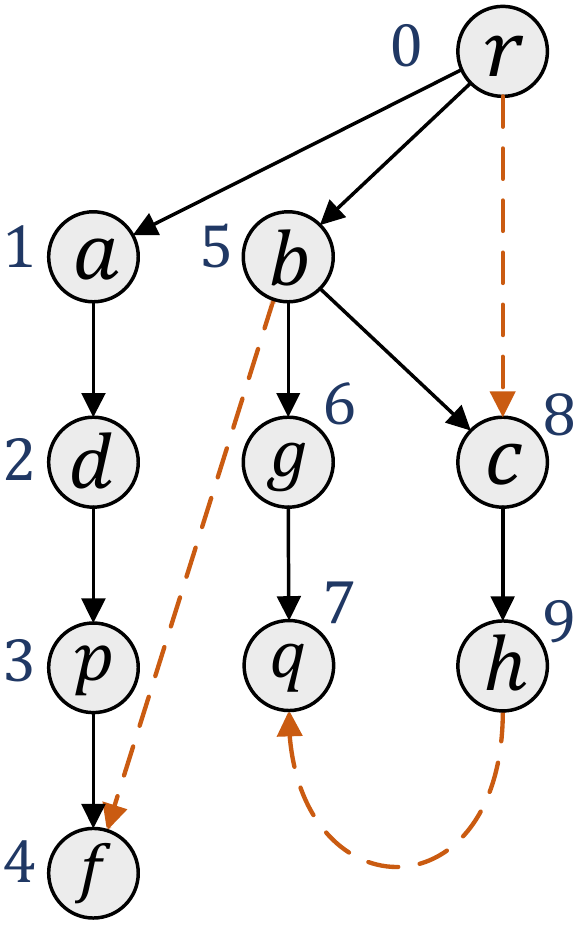}}
	
	\caption{We draw a graph $\mathcal{G}$ in four forms as shown in subfigures (a)-(d), where $\mathcal{T}_0,\mathcal{T}_1,\mathcal{T}_2$ and $\mathcal{T}_3$ are spanning trees of $\mathcal{G}$ constituted by solid lines. }
	\label{fig:basic_idea_example_G}
\end{figure}

To formally define our above observation~(Observation~\ref{obs:first}), we use a notation $\mathcal{C}(T,T_i)$, as defined in Definition~\ref{def:S(T,T_i)}.

\begin{definition}
\label{def:S(T,T_i)}
$\mathcal{C}(T,T_i)$ refers to a concrete number $k$, where (\romannumeral1) $T$ and $T_i$ are spanning trees of $G$; (\romannumeral2) $\forall u\in V(G)$, if $d\negthinspace f\negthinspace o(u,T)<k$, then $d\negthinspace f\negthinspace o(u,T)=d\negthinspace f\negthinspace o(u,T_i)$, otherwise $d\negthinspace f\negthinspace o(u,T)\neq d\negthinspace f\negthinspace o(u,T_i)$.
\end{definition} 

For instance, as depicted in Figure~\ref{fig:basic_idea_example_G}, $\mathcal{C}(\mathcal{T}_0,\mathcal{T}_1)=10$ and $\mathcal{C}(\mathcal{T}_0, \mathcal{T}_2)= 4$.

\begin{observation}\label{obs:first}
	$\mathcal{C}(T,T_i)\geq d\negthinspace f\negthinspace o(v,T_i)$, where $T_i$ is the DFS-Tree of the graph composed by $T$ and an edge batch $B_i=\{(u,v)\}$. 
\end{observation}
\begin{proof}
	As mentioned in Section~\ref{sec:preliminaries}, $T$ is restructured to $T_i$ based on Stipulation~\ref{sti:stipulation}. For one thing, if the edge $(u,v)$ in $\mathcal{B}_i$ is a forward cross edge, then $T$ is restructured to $T_i$ by adding an edge $(u,v)$ into $T$ and removing the edge $(w,v)$ from $T$, where $w$ is the parent node of $v$ on $T$. Thus, the above statement is valid. For another, if $(u,v)$ is not a forward cross edge, then the total depth-first order of $T$ and that of $T_i$ are the same. \hfill $\Box$
\end{proof}

Based on Observation~\ref{obs:first}, when $B_i$ refers to a subset of $E(G)$, we find out some nodes of $G$ could have the same depth-first orders on $T$ and on $T_i$, if their depth-first orders on $T$ are smaller $\Upsilon(T)$ as defined below. This observation is formally presented in Observation~\ref{obs:second}.

\begin{definition}
	\label{def:U(T)}
	$\Upsilon(T)=d\negthinspace f\negthinspace o(u,T)$, where (\romannumeral1) an edge $e$ in $G$ whose tail is $u$ is a forward cross edge as classified by $T$, and (\romannumeral2) for any other forward cross edge $e'$ in $G$ as classified by $T$, the depth-first order of the tail of $e'$ is larger than $d\negthinspace f\negthinspace o(u,T)$. If there is no forward cross edge in $G$ as classified by $T$ then $\Upsilon(T)=+\infty$.
\end{definition}

As depicted in Figure~\ref{fig:basic_idea_example_G},   $\Upsilon(\mathcal{T}_0)=\Upsilon(\mathcal{T}_1)=3$, $\Upsilon(\mathcal{T}_2)=6$ and $\Upsilon(\mathcal{T}_3)=+\infty$.

\begin{observation}
	\label{obs:second}
	$\forall B_i\subseteq E(G)$, $\mathcal{C}(T,T_i)\geq \Upsilon(T)$, where $T_i$ is the DFS-Tree of the graph composed by $T$ and $B_i$.
\end{observation}
\begin{proof}
	This statement is equivalent to ``$\forall u\in V(G)$, if $d\negthinspace f\negthinspace o(u,T)= k\leq \Upsilon(T)$, then $d\negthinspace f\negthinspace o(u,T_i)=k$'', which could be proved by contradiction.  Assuming there is a node $x$ in $G$ where $d\negthinspace f\negthinspace o(x,T)=k_x\leq\Upsilon(T)$, but $d\negthinspace f\negthinspace o(x,T_i)\neq k_x$. Since $d\negthinspace f\negthinspace o(x,T)\leq\Upsilon(T)$, according to Definition~\ref{def:U(T)}, $\forall e\in E(G)$, if $x$ is the tail or the head of $e$, then $e$ is not a forward cross edge as classified by $T$. Thus, if our assumption is correct, then there must be an edge $e'$ in $G$, where (1) $e'\in B_i$ is a forward cross edge as classified by $T$; (2) the depth-first order of the tail of $e'$ on $T$ must be smaller than $d\negthinspace f\negthinspace o(x,T)$. Otherwise, the assumption cannot be right, according to Stipulation~\ref{sti:stipulation}. However, if there is an edge $e'$ in $G$, $\Upsilon(T)$ must be smaller than the depth-first order of the tail of $e'$, which contradicts the premise.  \hfill $\Box$
\end{proof}

For example, given a batch of edges $B_i=\{e_1,e_2,(a,d)\}$, if $T$ is in the form of $\mathcal{T}_0$~(Figure~\ref{fig:basic_idea_example_G}(a)), $\mathcal{T}_2$ is the DFS-Tree of the graph composed by $T$ and $B_i$, where $\mathcal{T}_2$ is depicted in Figure~\ref{fig:basic_idea_example_G}(c).  $\mathcal{C}(\mathcal{T}_0,\mathcal{T}_2)=4\geq \Upsilon(\mathcal{T}_0)=3$.

Based on Observation~\ref{obs:second}, we find out that \textit{not all depth-first orders of the nodes on $T$ will keep changing, during the process of restructuring $T$ to a DFS-Tree of $G$, especially when we limit the elements that can be contained in $B_i$.}  A naive algorithm is presented, as shown in Algorithm~\ref{algo:naive_up_down}.

%\begin{minipage}[t]{0.8\linewidth}
	\begin{algorithm}[t]
		\caption{Naive EP-DFS$(G,r)$}
		\label{algo:naive_up_down}\small
		\begin{algorithmic}[1]
			\renewcommand{\algorithmicrequire}{\textbf{Input:}}
			\renewcommand{\algorithmicensure}{\textbf{Output:}}
			\renewcommand{\algorithmiccomment}[1]{  #1}
			
			\REQUIRE $G=(V,E)$ is an input graph; $r$ is a node of $G$ connected to all the other nodes in $G$.
			\ENSURE  A DFS-Tree of $G$.
			\STATE $T\leftarrow$ a spanning tree of $G$ rooted at $r$, $FNN\gets1$ \label{line:nepscc:initial_T}
			\STATE \textbf{while} $FNN<n$ \textbf{do}\label{line:nepscc:while_start}
			\STATE \quad  $\mathcal{B} \gets$Scanning$(G, FNN)$\label{line:nepscc:scannG}

			\STATE \quad Replacing $T$ with the DFS-Tree $T_\mathcal{B}$ of the graph composed by $T$ and $\mathcal{B}$\label{line:nepscc:ReplaceT}

			\STATE \quad $FNN\gets \mathcal{C}(T,T_\mathcal{B})$\label{line:nepscc:FNNget}
			\STATE \quad \textbf{if} $\mathcal{C}(T,T_\mathcal{B})>Max(\mathcal{B})+1$ \textbf{then} \label{line:nepscc:while_if}
			\STATE \quad \quad $F\negthinspace N\negthinspace N\gets Max(\mathcal{B})+1$\label{line:nepscc:if_FNN}
			\STATE \quad \textbf{end if}\label{line:nepscc:end_if}
			\STATE \textbf{end while}\label{line:nepscc:end_while}
			\STATE \textbf{return} $T$ \label{line:nepscc:return}	
		\end{algorithmic}
	\end{algorithm}
%\end{minipage}
%\vspace{20pt}

The naive algorithm restructures $T$ to a DFS-Tree of $G$ with iterations, which is a natural way to address the DFS problem on semi-external memory model. In each iteration, it sequentially scans $G$ once, in order to obtain a batch $\mathcal{B}$ of edges from $G$. Specifically, there is a parameter utilized in the naive EP-DFS, i.e. $FNN$ as defined below.

\begin{definition}
	\label{def:FNN}
	$FNN$~(Fixed Node Number) represents a concrete number, where, $\forall u\in V(T)$, if $d\negthinspace f\negthinspace o(u,T)<FNN$ in the $i$th iteration of EP-DFS~(Algorithm~\ref{algo:naive_up_down} and Algorithm~\ref{algo:up_down}), then the value of $d\negthinspace f\negthinspace o(u,T)$ will not be changed in any following iteration.
\end{definition}

Of course, during the whole process, we cannot change the root of $T$, so that $FNN$ is initialized to $1$, as demonstrated in Line~\ref{line:nepscc:initial_T}. 

Line~\ref{line:nepscc:scannG} loads an edge batch $\mathcal{B}$, as defined in Definition~\ref{def:B}. The reason is, in the $i$th iteration, if an edge $e$ related to a node $s$ whose depth-first order on $T$ is smaller than $FNN$, $e$ cannot be a forward cross edge as classified by $T$ in the $j$th iteration, where $j\geq i$. The above statement is proved by Theorem~\ref{theorem:reduction_edges}.

\begin{definition}
	\label{def:B}
	$\mathcal{B}$ is an edge batch of $G$, where 
	
	(\romannumeral1) $\mathcal{B}$ contains an edge $e=(u,v)$ of $G$, if $d\negthinspace f\negthinspace o(u,T)\geq FNN$ or $d\negthinspace f\negthinspace o(v,T)\geq FNN$; 
	
	(\romannumeral2) if $\exists e=(u,v)\in \mathcal{B}$ and the tail or the head of an edge $e'$ of $G$ has a depth-first order on $T$ that is no larger than $d\negthinspace f\negthinspace o(u,T)$ or $d\negthinspace f\negthinspace o(v,T)$, then $e'$ belongs to $\mathcal{B}$
	
	(\romannumeral3) there is an edge in $\mathcal{B}$ whose tail or head has a depth-first order on $T$ which is equal to $Max(\mathcal{B})$, and $\forall e=(u,v)\in \mathcal{B}$, either $d\negthinspace f\negthinspace o(u,T)\leq Max(\mathcal{B})$ or $d\negthinspace f\negthinspace o(v,T)\leq Max(\mathcal{B})$.
\end{definition}

As we limit the main memory maintains at most $2n$ edges of $G$, $\mathcal{B}$ has at most $n$ edges. After obtaining the edge batch $\mathcal{B}$ in Line~\ref{line:nepscc:scannG}, Line~\ref{line:nepscc:ReplaceT} restructures $T$ into the DFS-Tree $T_\mathcal{B}$ of the graph composed by $T$ and $\mathcal{B}$ according to Stipulation~\ref{sti:stipulation}. It could be proved in Theorem~\ref{theorem:reduction_edges} that parameter $FNN$ could be updated by the smallest value of $\mathcal{C}(T, T_\mathcal{B})$ and $Max(\mathcal{B})+1$, in case all the edges contained in $\mathcal{B}$ are not forward cross edges as classified by $T$.

\begin{theorem}\label{theorem:reduction_edges}
	$\Upsilon(T_\mathcal{B})\geq\, $Min$\big(\mathcal{C}(T,T_\mathcal{B}),Max(\mathcal{B})+1\big)$, if $\Upsilon(T)\geq FNN$ and $T_\mathcal{B}$ is the DFS-Tree of the graph composed by $T$ and an edge batch $\mathcal{B}$. Min$(i,j)$ represents the smaller value between $i$ and $j$.
\end{theorem}
\begin{proof}
	We prove the statement by contradiction, i.e. supposing $\Upsilon(T_\mathcal{B})< Min\big(\mathcal{C}(T,T_\mathcal{B}),Max(\mathcal{B})+1\big)$. In other words, $G$ contains a forward cross edge $e=(u,v)$ as classified by $T_\mathcal{B}$, where $d\negthinspace f\negthinspace o(u,T_\mathcal{B})<Min\big(\mathcal{C}(T,T_\mathcal{B}),Max(\mathcal{B})+1\big)$.  According to Observation~\ref{obs:second}, $d\negthinspace f\negthinspace o(v,T_\mathcal{B})>d\negthinspace f\negthinspace o(u,T_\mathcal{B})\geq FNN$. Hence, $FNN\leq d\negthinspace f\negthinspace o(u,T_\mathcal{B})<Min\big(\mathcal{C}(T,T_\mathcal{B})$, $Max(\mathcal{B})+1\big)$. Based on the premise of the statement, if a node $x$ has a depth-first order on $T$ that is smaller than $\mathcal{C}(T,T_\mathcal{B})$, then $d\negthinspace f\negthinspace o(x,T_\mathcal{B})=k$. As $d\negthinspace f\negthinspace o(u,T_\mathcal{B})<\mathcal{C}(T,T_\mathcal{B})$, there must be $d\negthinspace f\negthinspace o(u,T_\mathcal{B})=d\negthinspace f\negthinspace o(u,T)$ and both of them are smaller than $Min\big(\mathcal{C}(T,T_\mathcal{B}),Max(\mathcal{B})+1\big)$, i.e. $(u,v)\in\mathcal{B}$. Thus, after restructuring $T$ to $T_\mathcal{B}$, there is an edge in $\mathcal{B}$ that is a forward cross edge as classified by $T_\mathcal{B}$, which is impossible. \hfill$\Box$
\end{proof}

Example~\ref{exampe:naive_EP-DFS} takes one iteration of the naive EP-DFS for instance.

\begin{example}
	\label{exampe:naive_EP-DFS}
	Assuming at the beginning of one iteration of Algorithm~\ref{algo:naive_up_down}, $T$ is in the form of $\mathcal{T}_0$ shown in Figure~\ref{fig:basic_idea_example_G}(a), the value of $FNN$ is $1$. $\mathcal{B}=\{(r,a), (r,b), (r,c), (a,d), (d,p), (p,f),(b,f), (b,g), (b,c)\}$ so that $Max(\mathcal{B})=5$. That is, the DFS-Tree $T_\mathcal{B}$ of the graph composed by $\mathcal{T}$ and $\mathcal{B}$ is in the form of  $\mathcal{T}_2$ shown in Figure~\ref{fig:basic_idea_example_G}(c). Hence, $FNN$ is updated to $4$, as $\mathcal{C}(\mathcal{T}_0,\mathcal{T}_2)=4\leq Max(\mathcal{B})+1 = 6$.
\end{example}

\section{EP-DFS Algorithm} \label{sec:algorithm}

In practice, the implementation of Algorithm~\ref{algo:naive_up_down} is intricate, especially when $G$ is relatively large. (1) One reason is that at beginning the naive EP-DFS initializes $FNN$ to $1$. Even though, setting $FNN=1$ is correct, because during the whole process, the root of $T$ is fixed which is node $r$, and the depth-first order of leftmost child of $r$ is also fixed according to Stipulation~\ref{sti:stipulation}. This initialization of $FNN$ may let the naive EP-DFS scan the whole input graph many times meaninglessly. The reason is when $\Upsilon(T)\gg 1$, $G$ may contains a large number of edges, each of which has a tail or a head whose depth-first order on $T$ is in the range of $[1,\Upsilon(T))$, and these edges may not fit into one edge batch $\mathcal{B}$~($\mathcal{B}$ contians $n$ edges at most according to Section~\ref{sec:preliminaries}). (2) Another is, it is hard to load an edge batch $\mathcal{B}$, since the value of $Max(\mathcal{B})$ is unknown before $\mathcal{B}$ is obtained, and the distribution of the edges in $G$ is also unknown. (3) Furthermore, we find the naive EP-DFS may not be terminated when $G$ has more than $n$ edges whose tail or head is $u$, and $d\negthinspace f\negthinspace o(u,T)=FNN$. Because, in this case, according to Definition~\ref{def:B}, $\mathcal{B}$ cannot fit into the main memory under the restriction that at most $2n$ edges could be contained in the main memory.

To efficient restructure $T$ to a DFS-Tree of $G$, EP-DFS is presented in this section, and its pseudo-code is given in Algorithm~\ref{algo:up_down}. As shown in Lines~\ref{line:e+:initializeO}-\ref{line:e+:initializeOthers}, EP-DFS initializes parameter $FNN$ by a new method, which is named InitialRound and detailedly discussed in Section~\ref{sec:algo:FNNinitializatin}. Then, since loading an edge batch of $G$ with certain restrictions is hard and EP-DFS may not be terminated by using $\mathcal{B}$, Section~\ref{sec:algo:Bobtain} introduces the edge batch, named $\mathcal{B}^+$, and a light-weight index $\mathcal{N}^+\negthickspace$\textit{-index} to efficiently obtain $\mathcal{B}^+$. Section~\ref{sec:algo:FNNupdate} illustrates how to correctly update the value of $FNN$ as larger as possible when using $\mathcal{B}^+$ instead of $\mathcal{B}$. Section~\ref{sec:algo:optimization} presents an optimization algorithm for EP-DFS. Section~\ref{sec:algo:discussion} discusses the space consumption of EP-DFS and presents its implementation details. For the correctness of EP-DFS, please see Section~\ref{sec:algo:correctness}.

	\begin{algorithm}[t]
		\caption{EP-DFS$(G,r)$}
		\label{algo:up_down}\small
		\begin{algorithmic}[1]
			\renewcommand{\algorithmicrequire}{\textbf{Input:}}
			\renewcommand{\algorithmicensure}{\textbf{Output:}}
			\renewcommand{\algorithmiccomment}[1]{  #1}
			
			\REQUIRE $G=(V,E)$ is an input graph; $r$ is a node of $G$ connected to all the other nodes in $G$.
			\ENSURE  A DFS-Tree $T$.
			
			\STATE $T\gets$ a spanning tree of $G$ rooted at $r$  \label{line:star:T_initialize}
			
			\STATE $T,FNN\leftarrow$ InitialRound$(G,T)$\label{line:star:FNN} \hfill $\triangleright$ Procedure~\ref{algo:intialRound}
			
			\STATE $F_1\leftarrow FNN$  
			\STATE $\mathcal{N}^+\negthickspace$\textit{-index}$\,\leftarrow\,$Indexing($FNN,T,G$) \label{line:star:index} \hfill $\triangleright$ Procedure~\ref{algo:Indexing}
			\STATE \textbf{while} $FNN<n$ \textbf{do} \label{line:star:while_true} 
			
			\STATE \quad $\mathcal{B}^+\leftarrow\,$ObtainingEdges$(FNN$, $\mathcal{N}^+\negthickspace$\textit{-index}$, T)$ \label{line:star:initialmap}   \hfill $\triangleright$ Procedure~\ref{algo:ObtainingEdges}
			\STATE \quad Replace $T$ with the DFS-Tree $T_{\mathcal{B}^+}$ of the graph composed by $T$ and $B^+$
						
			\STATE \quad $F_2\leftarrow FNN$, $FNN\gets\mathcal{C}^+(T, T_{\mathcal{B}^+})$ \label{line:star:min1}
			\STATE \quad If $FNN>Max(\mathcal{B}^+)+1$ then $FNN=Max(\mathcal{B}^+)+1$\label{line:star:min2}
			
			\STATE \quad \textbf{if} $\frac{FNN-F_2}{n}\rightarrow0$ and $FNN-F_1>\gamma\times n$ \textbf{then} \label{line:star:while_if}
			
			\COMMENT{\quad\quad // restructuring and rewriting  $\mathcal{N}^+\negthickspace$\textit{-index}}
			
			\STATE \quad \quad $T,FNN\leftarrow$RoundI$\&$Reduction$(\mathcal{N}^+$\textit{-index}$, FNN, T)$,  $F_1\leftarrow FNN$  \label{line:star:roundIandReduction} 
			
			\STATE \quad \textbf{else if}  $\frac{FNN-F_2}{n}\rightarrow0$ \textbf{then} \label{line:star:while_else_if}
			\STATE \quad \quad $T,FNN\leftarrow$RoundI$(\mathcal{N}^+$\textit{-index}$, FNN, T)$ \label{line:star:roundI}
			
			\STATE \quad \textbf{end if}\label{line:star:while_if_end}

			\STATE \quad Rearrangement($T,FNN$) \label{line:star:rearrangement} \hfill $\triangleright$ Procedure~\ref{algo:Rearrangement}
			
			\STATE \textbf{end while}\label{line:star:end_while}
			\STATE \textbf{return} $T$ \label{line:star:return}
		\end{algorithmic}
	\end{algorithm}

\subsection{The initialization of $FNN$}\label{sec:algo:FNNinitializatin}

We present a procedure InitialRound$(G,T)$ in Procedure~\ref{algo:intialRound}, for the initialization of $T$ and that of parameter $FNN$. 

\setcounter{algorithm}{0}
	\begin{algorithm}[t]
		\floatname{algorithm}{Procedure}
		\caption{InitialRound$(G,T)$}
		\label{algo:intialRound}\small
		\begin{algorithmic}[1]
			\renewcommand{\algorithmicrequire}{\textbf{Input:}}
			\renewcommand{\algorithmicensure}{\textbf{Output:}}
			\renewcommand{\algorithmiccomment}[1]{  #1}

			\STATE \textbf{for} each $n$ edges of $G$ \textbf{do} \label{line:initialround:while_1_start}
			\STATE \quad Replace $T$ with the DFS-Tree of the graph composed by $T$ and these $n$ edges \label{line:initialround:while_1_dfs_in}
			\STATE \quad Rearrangement($T,0$) \label{line:initialround:while_1_rearrangement} \hfill $\triangleright$ Procedure~\ref{algo:Rearrangement}
			\STATE \textbf{end for} \label{line:initialround:end_while1}
			\STATE $T'\leftarrow T$ \label{line:initialround:T'}
	
			\STATE \textbf{for} each $n$ edges of $G$ \textbf{do} \label{line:initialround:while_2_start}
%			\STATE \quad sequentially load $n$ edges into memory \label{line:initialround:while_2_load}
			\STATE \quad Replace $T$ with the DFS-Tree of the graph composed by $T$ and these $n$ edges \label{line:initialround:while_2_dfs_in}
			\STATE \textbf{end for} \label{line:initialround:end_while2}
			\STATE \textbf{return} $T$ and $FNN\gets\mathcal{C}(T',T)$\label{line:initialround:end} \COMMENT{// $\Upsilon(T)\geq \mathcal{C}(T',T)$}	
		\end{algorithmic}
	\end{algorithm}

This procedure has two loops. Its first loop is utilized in Lines~\ref{line:initialround:while_1_start}-\ref{line:initialround:end_while1}. It scans $G$ once, as demonstrated in Line~\ref{line:initialround:while_1_start}. For each $n$ edges of $G$, it replaces the in-memory spanning tree $T$ of $G$ with the DFS-Tree of the graph composed by $T$ and these edges, as demonstrated in Line~\ref{line:initialround:while_1_dfs_in}. Moreover, in each iteration of this loop, we present a new algorithm to rearrange the orders of the nodes on $T$, which is named Rearrangement introduced in Procedure~\ref{algo:Rearrangement}. Comparing with the first loop, the second loop is simpler, in which we do not rearrange the nodes of $T$, as illustrated in Lines~\ref{line:initialround:while_2_start}-\ref{line:initialround:end_while2}. At the end of this function, Line~\ref{line:initialround:end} returns $T$, and sets $FNN$ to $\mathcal{C}(T',T)$, where $T'$ represents the form of $T$ obtained by the first loop~(Line~\ref{line:initialround:T'}). The correctness of setting $FNN$ to $\mathcal{C}(T',T)$ is proved in Theorem~\ref{theorem:min_initialize}, Section~\ref{sec:algo:correctness}.

%\begin{minipage}[t]{0.8\linewidth}
	\begin{algorithm}[t]
		\floatname{algorithm}{Procedure}
		\caption{Rearrangement$(T,FNN)$}
		\label{algo:Rearrangement}
		\begin{algorithmic}[1]\small
			\renewcommand{\algorithmicrequire}{\textbf{Input:}}
			\renewcommand{\algorithmicensure}{\textbf{Output:}}
			\renewcommand{\algorithmiccomment}[1]{  #1}
			
			\STATE \textbf{for} $i$ from $n-1$ to $FNN$ \textbf{do} \label{line:nodeweight:for} 
			\STATE \quad $u$ represents the node on  $T$ whose depth-first order is $i$ 
			\STATE \quad Assuming $u$ has $k$ children on $T$ which are $v_1,v_2,\dots,v_k$
			
			\STATE \quad $\mathcal{W}(u)\longleftarrow 1+\mathcal{W}(v_1)+\mathcal{W}(v_2)+\dots+\mathcal{W}(v_k)$\label{line:nodeweight:Wcompute}
			
			\STATE \quad \textbf{for} each integer $j$ in the range of $[0,\frac{k}{10^4}]$ \textbf{do} \label{line:nodeweight:for_2}
			\STATE \quad \quad $l\gets j\times 10^4$, QS-Rearrange($l+1,l+10^4$)\label{line:nodeweight:Sort}
			\STATE \quad \textbf{end for} \label{line:nodeweight:endfor2}
			\STATE \quad $l\gets \lfloor\frac{k}{10^4}\rfloor\times 10^4$, if $l<k$ then QS-Rearrange($l+1,k$)\label{line:nodeweight:addtionalSort}
			\STATE \textbf{end for}	\label{line:nodeweight:endforfinal}
		\end{algorithmic}
	\end{algorithm}
%\end{minipage}
%\vspace{15pt}

Our rearrangement algorithm is introduced in  Procedure~\ref{algo:Rearrangement}. It aims to rearrange the children of each node $u$ on $T$, according to their node weights. Specifically, the weight of $v$ is the size of $T_v$, denoted by $\mathcal{W}(v)$. $T_v$ is the subtree of $T$, where the root of $T_v$ is $v$ and leaves of $T_v$ are the leaves of $T$. Firstly, to compute the weights of nodes on $T$, we scan the nodes on $T$ in the reverse total depth-first order of $T$. When node $u$ on $T$ is scanned, its node weight $\mathcal{W}(u)$ is set to the value of $1+\mathcal{W}(v_1)+\mathcal{W}(v_2)+\dots+\mathcal{W}(v_k)$, as shown in Line~\ref{line:nodeweight:Wcompute} of Procedure~\ref{algo:Rearrangement}. Here, we assume $u$ has $k$ children on $T$ which are $v_1,v_2,\dots, v_k$. The correctness of this node-weight computation process is proved in Theorem~\ref{theorem:node_weitht}, Section~\ref{sec:algo:correctness}. Secondly, when node $u$ is scanned, we also rearrange its out-neighborhoods. However, considering that some nodes on $T$ may have a large set of out-neighbors, when $u$ has more than $10^4$ nodes on $T$, procedure Rearrangement sorts at most $10^4$ children of $u$ at a time, as shown in Lines~\ref{line:nodeweight:for_2}-\ref{line:nodeweight:addtionalSort}. Assuming $l_1$ and $l_2$ are two integers, where $l_1>l_2$. Each time a function QS-Rearrange is used to rearrange the $l_1$th to the $l_2$th out-neighbors of $u$ from the left to right, as shown in Line~\ref{line:nodeweight:Sort} and Line~\ref{line:nodeweight:addtionalSort}. It worth noting that, the sort algorithm used here is  quick-sort~\cite{DBLP:books/daglib/0037819}. Besides, if the rearranged $l_1$th to the rearranged $l_2$th out-neighbors of $u$ from the left to right are $v_{l_1},v_{l_1+1}\dots, v_{l_2}$, then $\mathcal{W}(v_{l_1})\geq\mathcal{W}(v_{l_1+1})\geq\dots\geq\mathcal{W}(v_{l_2})$.

Our rearrangement strategy is efficient, which is discussed below. For one thing, our rearrangement algorithm only traverse $T$ once, while traditional algorithms have to traverse $T$ at least twice, where the former is used to compute the weight of each node on $T$, and the latter is for sorting the children of each node on $T$ in a required order. Even though traversing an in-memory spanning tree is quite fast, it is still necessary to decrease the number of traversing $T$. That is because, in a semi-external algorithm, $T$ is changed without predictability so that it is also stored in the main memory with an unstructured data structure. To restructure an in-memory spanning tree to a DFS-Tree of $G$, a semi-external algorithm may have to rearrange $T$ many times, and the size of $G$ is often large. Note that, traversing $T$ in a semi-external algorithm does not involve any I/O costs since all the nodes are stored in the main memory and there is no input and output disk access. For another, traditional algorithms have to consume  $O(n\log n)$ time to rearrange the out-neighbors of nodes on $T$, since a node $u$ on $T$ may have up to $n$ out-neighbors. However, in EP-DFS, the time consumption of procedure Rearrangement is only $O(n)$, as discussed in the end of  Section~\ref{sec:algo:discussion}.
 
An example of our rearrangement algorithm is given below.

\begin{figure}[h]
	\centering
	\renewcommand{\thesubfigure}{}
	\subfigure[(a)]{
		\includegraphics[scale = 0.35]{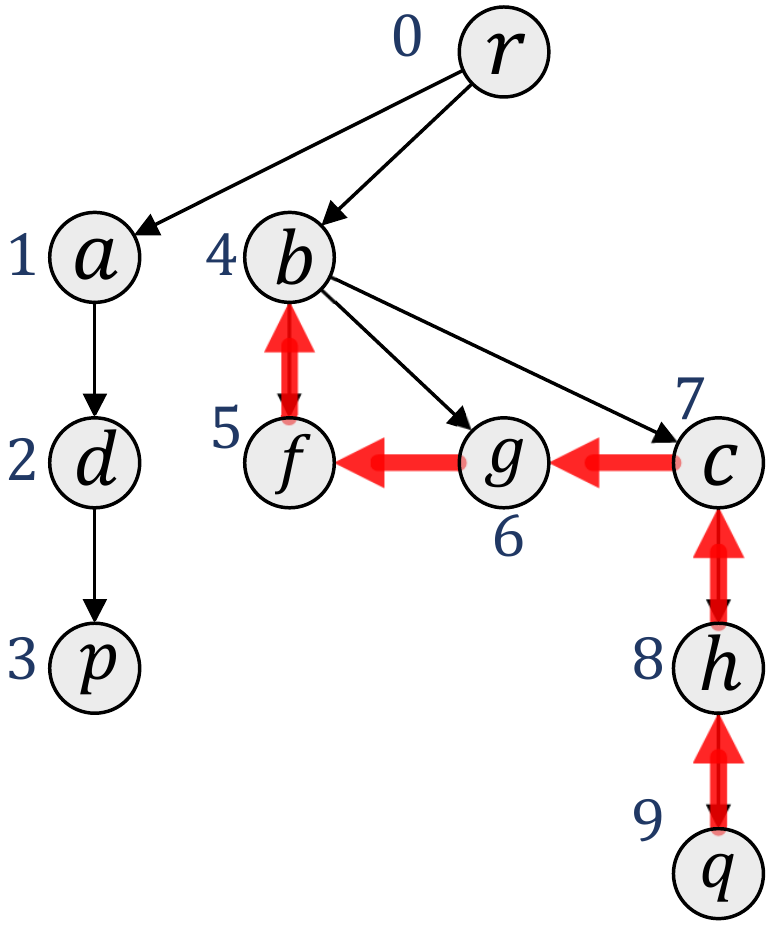}}
	\quad
	\subfigure[(b)]{
		\includegraphics[scale = 0.35]{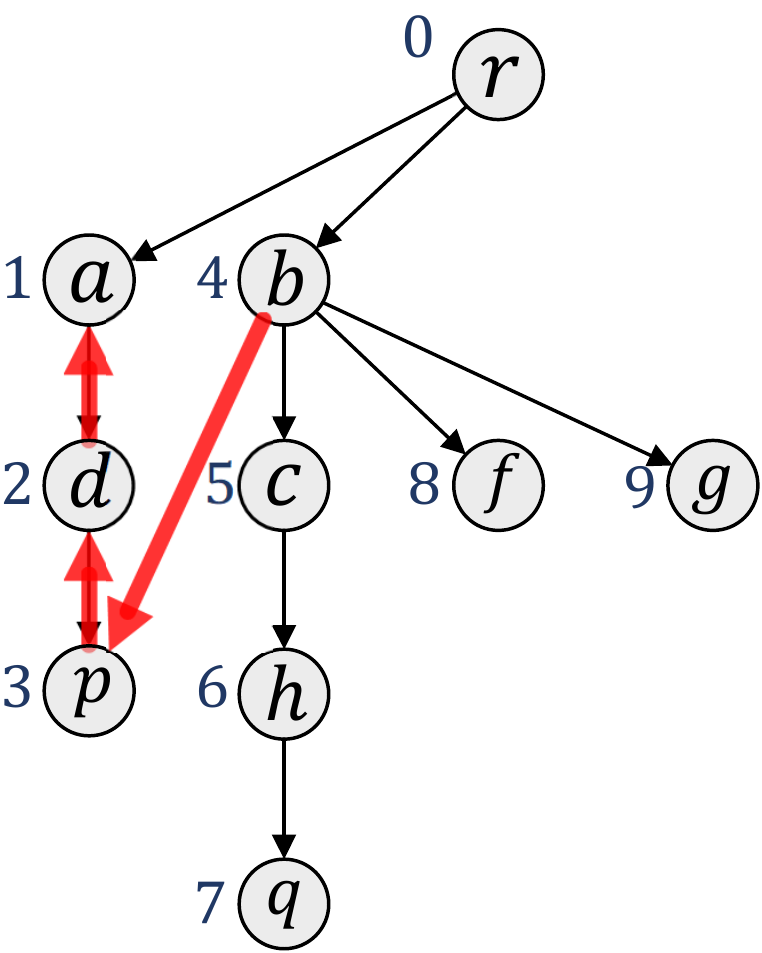}}
	\quad
	\subfigure[(c)]{
		\includegraphics[scale = 0.35]{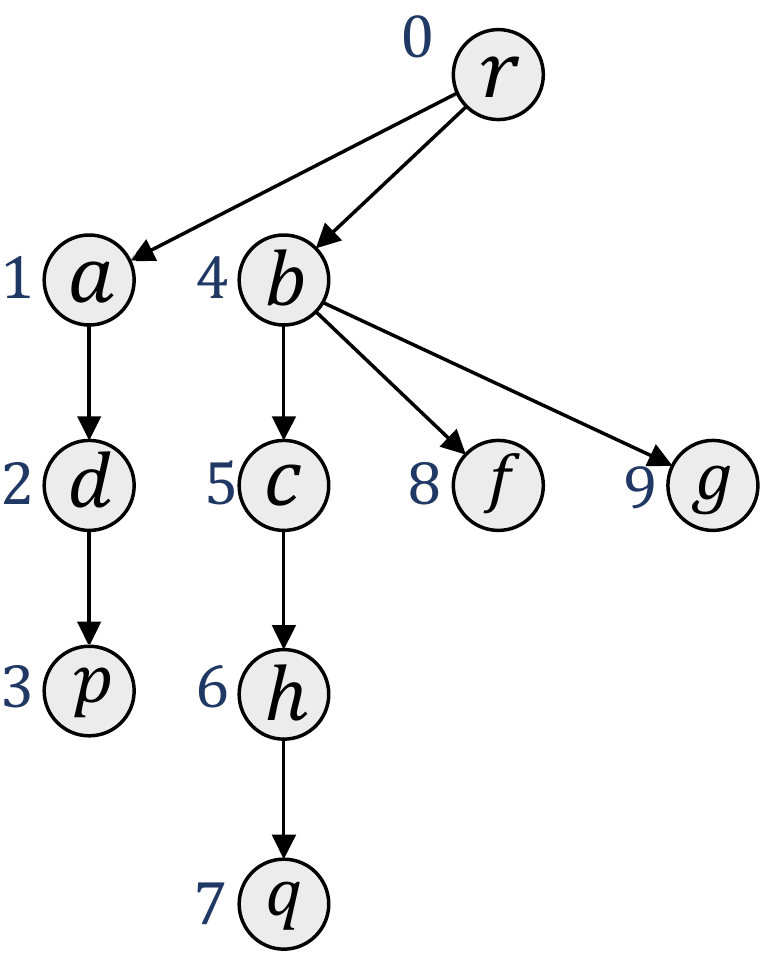}}
	\caption{An example of Procedure~\ref{algo:Rearrangement}, where the input spanning tree of $G$ is in the form of $\mathcal{T}_1$ shown in Figure~\ref{fig:basic_idea_example_G}(b), and the value of the input parameter $FNN$ is $1$. }
	\label{fig:rearrangement_example}
\end{figure}

\begin{example}
	Figure~\ref{fig:rearrangement_example} is a vivid schematic view of how the rearrangement strategy works, when $T$ is in the form of $\mathcal{T}_1$~(Figure~\ref{fig:basic_idea_example_G}(b)), and the input $FNN$ is set to $1$. Figure~\ref{fig:rearrangement_example}(a) shows Procedure~\ref{algo:Rearrangement} first restructures the out-neighborhoods of nodes $q,h,c,\dots, b$. Figure~\ref{fig:rearrangement_example}(b) illustrates that the last node whose out-neighborhood is rearranged by Procedure~\ref{algo:Rearrangement} is $a$ not $r$, since $FNN$ is set to $1$. Figure~\ref{fig:rearrangement_example}(c) depicts the rearranged $T$.
\end{example}

\subsection{How to efficiently obtain an edge batch} \label{sec:algo:Bobtain}
	As mentioned below, in the naive EP-DFS, restructuring $T$ with an edge set $\mathcal{B}$~(defined in Definition~\ref{def:B}) is impractical, since it may cause the naive algorithm cannot be terminated. To be specific, when $G$ has more than $n$ edges related to the node whose depth-first order on $T$ is $FNN$, the edge batch $\mathcal{B}$ in the naive algorithm cannot fit into the main memory under our restriction~(Section~\ref{sec:preliminaries}). Moreover, since the edge entries of $G$ is stored on disk with an unknown order, it is quite difficult to obtain an edge batch of $G$ with certain requirements.

	In order to address the problem that the edge batch $\mathcal{B}$ may not be loaded in memory under the restriction that at most $2n$ edges of $G$ could be maintained in the main memory, in EP-DFS, $\mathcal{B}^+$ is used instead of $\mathcal{B}$.
	
	\begin{definition}
		\label{def:B^+}
		$\mathcal{B}^+$ is an edge batch of $G$, where 
		
		(\romannumeral1) $\mathcal{B}^+$ contains an edge $e=(u,v)$ of $G$ if $FNN\leq d\negthinspace f\negthinspace o(u,T)$ and $d\negthinspace f\negthinspace o(v,T)>FNN$; 
		
		(\romannumeral2) $e'=(u',v')\in \mathcal{B}^+$, if  $\exists e=(u,v)\in \mathcal{B}^+$ and $d\negthinspace f\negthinspace o(u',T)\leq d\negthinspace f\negthinspace o(u,T)$; 
		
		(\romannumeral3) $\exists e=(u,v)$, $d\negthinspace f\negthinspace o(u,T)=Max(\mathcal{B}^+)$, and $\forall e'=(u',v')\in \mathcal{B}^+$, $d\negthinspace f\negthinspace o(u',T)\leq Max(\mathcal{B}^+)$.
	\end{definition}

	$\forall u\in V(T)$ such that $FNN\leq d\negthinspace f\negthinspace o(u,T)\leq Max(\mathcal{B}^+)$,  $\mathcal{B}^+$ contains only the outgoing edges $(u,v)$ of $u$ where $d\negthinspace f\negthinspace o(v,T)$ is larger than $FNN$. Thus, for each node $v$ in $G$, as the number of $v$'s outgoing edges is less than $n$ discussed in Section~\ref{sec:preliminaries}, there is an edge batch $\mathcal{B}^+$ that  $|\mathcal{B}^+|\leq n$.  As an example, given $FNN=6$ and $Max(\mathcal{B}^+)=7$, $\mathcal{B}^+=\{(g,q),(c,h)\}$, in Figure~\ref{fig:basic_idea_example_G}(c).
		
	In order to efficiently access $G$ and obtain the edge batch $\mathcal{B}^+$, we devise a lightweight index, named $\mathcal{N}^+\negthickspace$\textit{-index}. Its indexing algorithm is presented in Procedure~\ref{algo:Indexing}.
	
	\begin{algorithm}[h]
		\floatname{algorithm}{Procedure}
		\caption{Indexing$(FNN,T,G)$}
		\label{algo:Indexing}\small
		\begin{algorithmic}[1]
			\renewcommand{\algorithmicrequire}{\textbf{Input:}}
			\renewcommand{\algorithmicensure}{\textbf{Output:}}
			\renewcommand{\algorithmiccomment}[1]{  #1}
			
			\STATE $i\gets0$, $\mathcal{E}_i\gets$ an empty edge list \label{line:indexing:first_line}
			\STATE \textbf{for} each edge $e=(u,v)$ in $G$ \textbf{do} \label{line:indexing:for}
			\STATE \quad \textbf{if} $d\negthinspace f\negthinspace o(u,T)<FNN$ \textbf{or} $d\negthinspace f\negthinspace o(v,T)\leq FNN$ \textbf{then} \label{line:indexing:for_if1}
			\STATE \quad \quad \textbf{continue} \label{line:indexing:for_if1_continue}
			\STATE \quad \textbf{end if} \label{line:indexing:for_end_if1}
			\STATE \quad $\mathcal{E}_i\gets \mathcal{E}_i\cup e$\label{line:indexing:E_i_cup}
			\STATE \quad \textbf{if} $\mathcal{E}_i$ cannot be enlarged any more \textbf{then}\label{line:indexing:for_if2}
			\STATE \quad \quad Sort the edges in $\mathcal{E}_i$, and then stored them on disk \label{line:indexing:for_if2_sort}
			\STATE \quad \quad $i\gets i+1$, $\mathcal{E}_i\gets$ an empty edge list\label{line:indexing:for_if2_iadd}
			\STATE \quad \textbf{end if}\label{line:indexing:for_end_if2}
			\STATE \textbf{end for}\label{line:indexing:for_end}
			\STATE Sort the edges in $\mathcal{E}_i$\label{line:indexing:sort}
			\STATE Obtain an edge stream $\mathcal{S}_v$ by merging all the ordered edge lists $\mathcal{E}_0$, $\mathcal{E}_1, \dots, \mathcal{E}_i$\label{line:indexing:S_v}
			\STATE \textbf{return} $\mathcal{N}^+\negthickspace$\textit{-index} by compressing $\mathcal{S}_v$\label{line:indexing:return}	 	
		\end{algorithmic}
	\end{algorithm}
	
	To construct $\mathcal{N}^+\negthickspace$\textit{-index}, the edges $e$$=$$(u,v)$ of $G$ are sequentially scanned in Lines~\ref{line:indexing:for}-\ref{line:indexing:for_end}. When $d\negthinspace f\negthinspace o(u,T)$ is smaller than $FNN$ or $d\negthinspace f\negthinspace o(v,T)$ is no larger than $FNN$, $e$ is discarded directly. That is because $FNN<\Upsilon(T)$, based on Theorem~\ref{theorem:min_initialize}. Hence, $\forall (u,v)\in E(G)$, if $d\negthinspace f\negthinspace o(u,T)< FNN$ and $d\negthinspace f\negthinspace o(v,T)\leq FNN$, then $e$ cannot be a forward cross edge as classified by $T$ in Procedure~\ref{algo:Indexing}. 
	
	Otherwise, Line~\ref{line:indexing:E_i_cup} adds $e$ into an edge list $\mathcal{E}_i$, where $\mathcal{E}_i$ is initialized in Line~\ref{line:indexing:first_line} and Line~\ref{line:indexing:for_if2_iadd}. When the edge list $\mathcal{E}_i$ cannot be enlarged anymore~(Line~\ref{line:indexing:for_if2}), we sort the edges in $\mathcal{E}$ and store them on disk, as demonstrated in Line~\ref{line:indexing:for_if2_sort}. Here, the edges in $\mathcal{E}_i$ are ordered by the following rules. (1) Assuming the nodes of $G$ are $u_1, u_2, \dots, u_j,\dots, u_{n}$. (2) $\mathcal{E}_i(u_j)=[]$, if $\forall v\in V(G)$, there is no edge $(u_j,v)$ in $\mathcal{E}_i$; otherwise $\mathcal{E}_i(u_j)=[(u_j,v_1), (u_j,v_2),\dots, (u_j,v_{k})]$, iff, $\forall l\in[i,k]$, $(u_j,v_l)\in \mathcal{E}_i$. (3) The ordered $\mathcal{E}_i$ is composed of $\mathcal{E}_i(u_1),\mathcal{E}_i(u_2),\dots$, and $\mathcal{E}_i(u_p)$ connected end to end. After scanning $G$, an edge stream $\mathcal{S}_v$ is obtained by a multi-way merge of all the ordered edge lists $\mathcal{E}_0, \mathcal{E}_1, \dots, \mathcal{E}_i$, as illustrated in Line~\ref{line:indexing:S_v}. For clarity, an instance is given in Example~\ref{example:indexing}.
	
	\begin{example}
		\label{example:indexing}
		Assuming the input graph $G$ is constituted by all the solid and dotted lines shown in Figure~\ref{fig:indexingExample}(a). One of its spanning trees $T$ is constituted by all the solid lines demonstrated in Figure~\ref{fig:indexingExample}(a). The given value of $FNN$ is $6$. According to Procedure~\ref{algo:Indexing}, given an edge $e=(u,v)$ of $G$, if $d\negthinspace f\negthinspace o(u,T)<FNN$ or $d\negthinspace f\negthinspace o(v,T)\leq FNN$, then $e$ is discarded. Only a small set of edges in $E(G)$ could be used to construct $\mathcal{N}^+\negthickspace$\textit{-index}, which are depicted in Figure~\ref{fig:indexingExample}(b). Here, to illustrated the following process of Procedure~\ref{algo:Indexing}, we assume the main memory sorts $3$ edges at a time. Then, all these edges in Figure~\ref{fig:indexingExample}(b) are divided into three edge lists, where these unordered edge lists are demonstrated in Figure~\ref{fig:indexingExample}(c) and their ordered forms are given in Figure~\ref{fig:indexingExample}(d). After using external sort to merge all these ordered edge lists, an edge streaming can be obtained, as shown in Figure~\ref{fig:indexingExample}(e).
	\end{example}

	\begin{figure}[t]
		\centerline{\includegraphics[scale = 0.4]{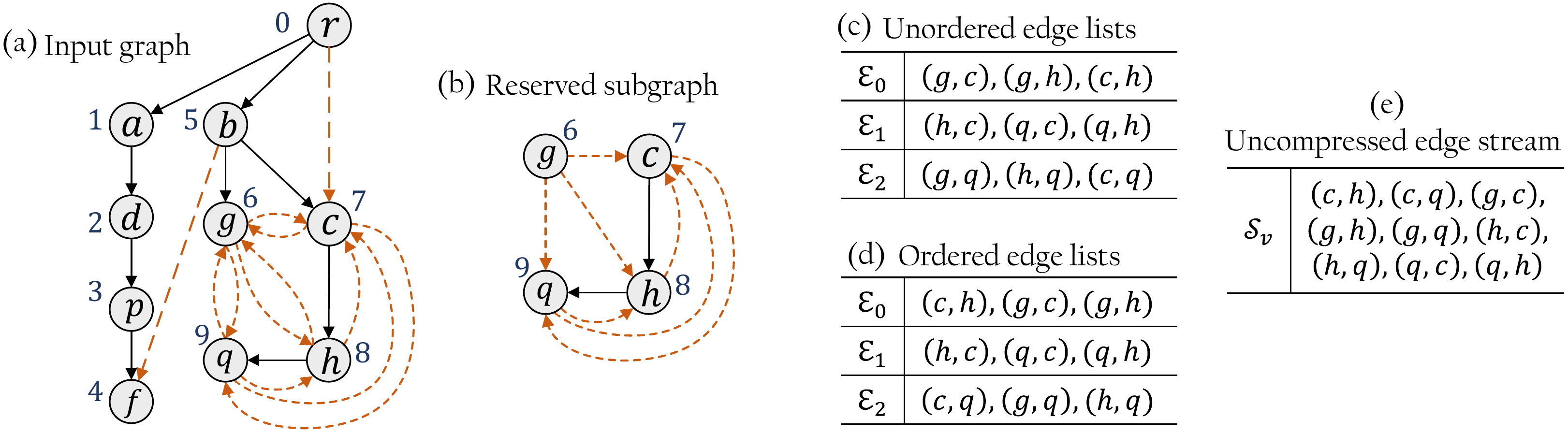}}
		\caption{An example of Procedure~\ref{algo:Indexing}, where (\romannumeral1) the input graph is constituted by all the edges demonstrated in subgraph (a), (\romannumeral2) the input spanning tree is constituted by the black solid edges of subgraph (a), and (\romannumeral3) the value of $FNN$ is $6$.}
		\label{fig:indexingExample}
	\end{figure}

	Finally, $\mathcal{N}^+\negthickspace$\textit{-index} could be obtained by compressing $\mathcal{S}_v$. Here, the compression algorithm used in this paper is \cite{BoVWFI,DBLP:conf/dcc/BoldiV04}, which could achieve the best compression rates~(about 2-3 bits per link), as far as we know. In order to utilize our $\mathcal{N}^+\negthickspace$\textit{-index}, at least two attributes should be maintained for each node $u$ on $T$, which are $u.OF$ and $u.OD$. The former $u.OF$ is the \textit{offset value}~\cite{DBLP:conf/dcc/BoldiV04} of the node $u$ in $\mathcal{N}^+\negthickspace$\textit{-index}. The latter $u.OD$ represents the out-degree of node $u$ on the graph composed by the edges in $\mathcal{N}^+\negthickspace$\textit{-index}.

	Based on $\mathcal{N}^+\negthickspace$\textit{-index}, we construct the edge batch $\mathcal{B}^+$ by procedure ObtainingEdges$(FNN,\mathcal{N}^+$\textit{-index}$,T)$, as shown in Procedure~\ref{algo:ObtainingEdges}.  It is presented for obtaining the edge batch $\mathcal{B}^+$ with a given variable $FNN$ on index $\mathcal{N}^+\negthickspace$\textit{-index}, and for avoiding random disk accesses.

	\begin{algorithm}[t]
		\floatname{algorithm}{Procedure}
		\caption{ObtainingEdges$(FNN ,\mathcal{N}^+$\textit{-index}$, T)$}
		\label{algo:ObtainingEdges}
		\small
		\begin{algorithmic}[1]
			\renewcommand{\algorithmicrequire}{\textbf{Input:}}
			\renewcommand{\algorithmicensure}{\textbf{Output:}}
			\renewcommand{\algorithmiccomment}[1]{  #1}
			
			\STATE $\mathcal{O}\leftarrow$ an empty offset list  \label{line:e+:initializeO}
			\STATE $\kappa\leftarrow0$, $Max(\mathcal{B}^+)\leftarrow FNN$ and $\mathcal{B}^+\leftarrow$ an empty edge set \label{line:e+:initializeOthers}
			\STATE \textbf{while} $Max(\mathcal{B}^+)<n$ \textbf{then} \label{line:e+:while}
			\STATE \quad $u$ is the node on $T$ whose depth-first order is $Max(\mathcal{B}^+)$  \label{line:e+:while_u}
			\STATE \quad \textbf{if} $\kappa+u.OD>n$ \textbf{then} \label{line:e+:while_if_1}
			\STATE \quad \quad $\mathcal{B}^+ \leftarrow \mathcal{B}^+\cup\,$LoadSequentially$(\mathcal{N}^+$\textit{-index},$\mathcal{O})$ \label{line:e+:while_if_1_load}
			\STATE \quad \quad $\mathcal{O}\leftarrow$ an empty offset list \label{line:e+:while_if_1_initializeO}, $\kappa\leftarrow|\mathcal{B}^+|$ \label{line:e+:while_if_1_kappa}
			\STATE \quad \textbf{end if}
			\STATE \quad \textbf{if} $\kappa+u.OD>n$ \textbf{then} \label{line:e+:while_if_2}
			\STATE \quad \quad \textbf{break} \label{line:e+:while_if_2_break}
			\STATE \quad \textbf{end if}
			\STATE \quad \textbf{if} $u.OD\neq0$ \textbf{then} \label{line:e+:while_if_3}
			\STATE \quad \quad $\mathcal{O}=\mathcal{O}\cup \{u.OF\}$ \label{line:e+:while_if_3_O},  $\kappa=\kappa+u.OD$ \label{line:e+:while_if_3_kappa}
			\STATE \quad \textbf{end if}
			\STATE \quad $Max(\mathcal{B}^+)\leftarrow Max(\mathcal{B}^+)+1$ \label{line:e+:while_eta_add_1}
			\STATE \textbf{end while}\label{line:e+:while_end}
			\STATE \textbf{return} $\mathcal{B}^+\cup\,$LoadSequentially$(\mathcal{N}^+$\textit{-index},$\mathcal{O})$ \label{line:e+:return}	
		\end{algorithmic}
	\end{algorithm}

First of all, Lines~\ref{line:e+:initializeO}-\ref{line:e+:initializeOthers} initialize $\mathcal{O}$ and $\mathcal{B}^+$ as an empty edge set and an empty offset list, respectively. $Max(\mathcal{B}^+)=FNN$, and $\kappa=0$. A loop, in Lines~\ref{line:e+:while}-\ref{line:e+:while_end}, runs until the value of $Max(\mathcal{B}^+)$ is no smaller than $n$, or the sum of  $\kappa$ and $u.OD$ is larger than $n$, as shown in Line~\ref{line:e+:while_if_2}. Here, $u$ represents the node on $T$ whose depth-first order on $T$ is equal to $Max(\mathcal{B}^+)$, as demonstrated in Line~\ref{line:e+:while_u}. When the sum of $\kappa$ an $u.OD$ is larger than $n$, as illustrated in Lines~\ref{line:e+:while_if_1}-\ref{line:e+:while_if_1_kappa}, we load the edges that are related to the offsets in $\mathcal{O}$ sequentially to reduce random disk seek operations. Besides, we reset $\mathcal{O}$ and $\kappa$ to an empty offset list and $|\mathcal{B}^+|$, respectively. If $u.OD$ is not equal to $0$ in Line~\ref{line:e+:while_if_3}, then $\mathcal{O}=\mathcal{O}\cup \{u.OF\}$ and $\kappa=\kappa+u.OD$~($\kappa<n$, according to Line~\ref{line:e+:while_if_2}). At the end of each iteration of this loop, $Max(\mathcal{B}^+)$ is updated to $Max(\mathcal{B}^+)+1$ as shown in Line~\ref{line:e+:while_eta_add_1}.

\subsection{How to update $FNN$} \label{sec:algo:FNNupdate}

In the naive EP-DFS, after replacing $T$ with the DFS-Tree $T_\mathcal{B}$ of the graph composed by $T$ and the edge batch $\mathcal{B}$~(Definition~\ref{def:B}), the value of parameter $FNN$ is updated by the smaller value of $\mathcal{C}(T,T_\mathcal{B})$ and $Max(\mathcal{B}^+)$. Even though the way that the naive EP-DFS updates parameter $FNN$ is correct, the difference between the value of $\Upsilon(T)$~(Definition~\ref{def:U(T)}) and the updated $FNN$ in the naive EP-DFS is still large. Since our EP-DFS can be terminated only when the value of $FNN$ is no smaller than $n$, we have put in a lot of efforts to further increase the value of $FNN$, after replacing $T$ with $T_{\mathcal{B}^+}$. Here, $T_{\mathcal{B}^+}$ represents the DFS-Tree of the graph composed by $T$ and the edge batch $\mathcal{B}^+$~(Definiton~\ref{def:B^+}).

Our current result~(Theorem~\ref{theorem:E+}) shows that, we could set the value of $FNN$ to the smaller value of $\mathcal{C}^+(T, T_{\mathcal{B}^+})$ and $Max(\mathcal{B}^+)$, as defined in Definition~\ref{def:S+(T, T)}.
\begin{definition}
	\label{def:S+(T, T)}
	$\mathcal{C}^+(T, T_{\mathcal{B}^+})=d\negthinspace f\negthinspace o(u, T_{\mathcal{B}^+})$, iff, (\romannumeral1) $d\negthinspace f\negthinspace o(u, T)> Max(\mathcal{B}^+)$, and (\romannumeral2) $\forall v\in V(G)$, if $d\negthinspace f\negthinspace o(v, T)> Max(\mathcal{B}^+)$ and $v\neq u$, then $d\negthinspace f\negthinspace o(v,T_{\mathcal{B}^+})>d\negthinspace f\negthinspace o(u, T_{\mathcal{B}^+})$.
\end{definition}

Example~\ref{example:EP-DFS} is an instance of one iteration of EP-DFS.

\begin{example}
	\label{example:EP-DFS}
	Given $T$ in the form of $\mathcal{T}_0$ shown in Figure~\ref{fig:basic_idea_example_G}(a). Assuming $FNN = 0$, and $\mathcal{B}^+=\{(r,a),(r,b),(r,c),$ $(a,d),(d,p),(p,f),(b,f),(b,g),(b,c)\}$. That is $Max(\mathcal{B}^+)=5$, $\mathcal{T}_2$~(Figure~\ref{fig:basic_idea_example_G}(c)) is the DFS-Tree of the graph composed by $T$ and $\mathcal{B}^+$, and $\mathcal{C}^+(\mathcal{T}_0,\mathcal{T}_2)=6$. Hence, $FNN$ will be updated to $6$.
\end{example}

\subsection{Optimization}\label{sec:algo:optimization}

In order to reduce the iteration times of EP-DFS, an optimization is devised in this part. Its pseudo-code is presented in Lines~\ref{line:star:while_if}-\ref{line:star:while_if_end} of Algorithm~\ref{algo:up_down}. For clarity, notation $F_1$ is used to denote the value of $FNN$ initialized in Line~\ref{line:star:initialmap} or updated in Line~\ref{line:star:roundIandReduction}. Notation $F_2$ is used to denote the value of $FNN$ before it is updated in Lines~\ref{line:star:min1}-\ref{line:star:min2}. $\gamma$ is a threshold for determining when to restructure $\mathcal{N}^+\negthickspace$\textit{-index}, which normally will be set to $10\%$, in Line~\ref{line:star:while_if}. $\frac{FNN-F_2}{n}\to 0$ indicates that the difference between the total depth-first orders of $T$ and $T_{\mathcal{B}^+}$ is small. 

\textit{RoundI$(\mathcal{N}^+\negthickspace$\textit{-index}$, FNN, T)$.} This procedure of EP-DFS is used to restructure $T$ with all the edges contained in $\mathcal{N}^+$\textit{-index}, by the following way. First of all, it loads the edges contained in $\mathcal{N}^+$\textit{-index} by batches, and each batch contains at most $n$ edges. For an edge $e=(u,v)$ in this index, $e$ is loaded into the main memory, iff, the depth-first order of $u$ on $T$ is no smaller than $FNN$ and the depth-first order of  $v$ on $T$ is larger than $FNN$. Second, it executes function Rearrangement after every five invocations of DFS. Assuming the input spanning tree $T$ is in the form of $T_{in}$ and the output spanning tree $T$ is in the form of $T_{out}$, then RoundI updates $FNN$ to $\mathcal{C}(T_{in},T_{out})$.

\textit{RoundI}$\&$\textit{Reduction$(\mathcal{N}^+\negthickspace$\textit{-index}$, FNN, T)$.} This procedure is similar to procedure RoundI, which uses edges in $\mathcal{N}^+\negthickspace$\textit{-index} to restructure $T$, and updates $FNN$ to $\mathcal{C}(T_{in},T_{out})$ as discussed above.  In addition to that, it also restructures $\mathcal{N}^+\negthickspace$\textit{-index} with the way of Procedure~\ref{algo:Indexing}. To be specific, it scans $\mathcal{N}^+\negthickspace$\textit{-index} sequentially. For each edge $e=(u,v)$ contained in $\mathcal{N}^+\negthickspace$\textit{-index}, it discards $e$ if $d\negthinspace f\negthinspace o(u,T)<FNN$ or $d\negthinspace f\negthinspace o(v,T)\leq FNN$, since $e$ cannot be a forward cross edge as classified by $T$ in the following iterations of EP-DFS. The rest of the edges, which are not discarded, are processed by batches: $\mathcal{E}_0, \mathcal{E}_1,\dots, \mathcal{E}_i, \dots$. For an edge batch $\mathcal{E}_i$, $T$ is replaced by the DFS-Tree of the graph composed by $T$ and $\mathcal{E}_i$. Then, $\mathcal{E}_i$ is ordered in the same way that $\mathcal{E}_i$ is ordered in Procedure~\ref{algo:Indexing}, an example of which is shown in Example~\ref{example:indexing}. After all the edge bathes are processed, an edge stream $\mathcal{S}_v$ is obtained by merging all the ordered edge batches~(lists), with external sort. Then, $\mathcal{N}^+\negthickspace$\textit{-index} is replaced by the compressed $\mathcal{S}_v$.

\subsection{Discussion and Implementation details} \label{sec:algo:discussion}
%In this section, we present a discussion about EP-DFS and its implementation details.

Compared with traditional algorithms~(Section~\ref{sec:related_works}), EP-DFS requires  \textit{simpler CPU calculation}, \textit{fewer random disk accesses} and \textit{lower memory space consumption}. The reason is as follows. Firstly, EP-DFS prunes the edges of the input graph $G$ efficiently, and only based on the total depth-first order of the nodes on $T$. Secondly, EP-DFS accesses the input disk-resident graph only sequentially; EP-DFS accesses the edges in $\mathcal{N}^+\negthickspace$\textit{-index} either sequentially or ObtainingEdges. Third, EP-DFS only needs to hold $2n$ edges of $G$ in the main memory, and keeps $3$ attributes for each node on $T$, i.e. its depth-first order, $u.OF$ and $u.OD$. 

Our implementation method for EP-DFS is presented below, which could protect it from being affected by the garbage collection mechanism of the implementation language, e.g. C\# and java. We assume each node of the input graph $G$ could be represented by a $32$-bit integer. Then, an integer array $\mathbb{A}_1$ of length $n$ is used for maintaining the node list of $G$. And an integer array $\mathbb{A}_2$ of length $4 n$ is used for maintaining $2 n$ edges of $G$, where a half is related to $T$ while the others correspond to $\mathcal{B}^+$. Plus, an integer array $\mathbb{A}_3$ of length $3 n$ is used for maintaining node attributes. Since we assume that each node could be represented as a $32$-bit integer, it is obvious that $d\negthinspace f\negthinspace o(u,T)$ and $u.OD$ can be represented as a $32$-bit integer. The reason why $u.OF$ can also be denoted as a $32$-bit integer is discussed in \cite{BoVWFI}.

The in-memory spanning tree $T$ of $G$ and one edge batch are maintained in the main memory by arrays $\mathbb{A}_1$ and $\mathbb{A}_2$. The elements in $\mathbb{A}_2$ are organized as one-way linkedlists. They may represent (\romannumeral1) the unused memory space, (\romannumeral2) the out-neighborhoods of nodes on $T$, and (\romannumeral3) a stack discussed later. We let $\mathbb{A}_1[i]$ to denote the $i$th element of $\mathbb{A}_1$ and let $\mathbb{A}_2[i]$ to denote the $i$th element of $\mathbb{A}_2$. Assuming $\mathbb{A}_1[i]=k$, and node $v_i$, $v_j$ and $v_l$ are the $i$th node, the $j$th node and the $l$th node of $G$, respectively. Then, (\romannumeral1) $\mathbb{A}_2[2k+1]=j$ represents the rightmost child of $v_i$ is $v_j$, and (\romannumeral2) if $\mathbb{A}_2[2k]=l\in[0,2n]$, then $v_l$ is the left brother of $v_j$. Thus, every two integers of $\mathbb{A}_2$ are used to represent an edge $(u,v)$, and each integer of $\mathbb{A}_1$ is used to denote the index of an edge in $\mathbb{A}_2$ whose head is the rightmost child of its tail. To be more specific, an example of how these two arrays work $\mathbb{A}_1$ and $\mathbb{A}_2$ in EP-DFS is given in Example~\ref{example:implementation}.

\begin{example}
	\label{example:implementation}
	Supposing $\mathcal{G}$ of Figure~\ref{fig:basic_idea_example_G} is the input graph. The nodes $r,a,b,c,d,f,g,h,p,q$ in $G_1$ are mapped into integers $0,1,2,3,4,5,6,7,8,9$. An array $\mathbb{A}_1$ of length $10$ is initialized, and an array $\mathbb{A}_2$ of length $10\times4$ is also initialized. The initialized arrays are in from (1) of Table~\ref{tab:implementation}. Then, firstly, we load the spanning $\mathcal{T}_0$, as shown in Figure~\ref{fig:basic_idea_example_G}(a), into the main memory, when these two arrays are in form (2) of Table~\ref{tab:implementation}. Secondly, we add an edge batch $\{e_1,e_2,e_3\}$ as shown in Figure~\ref{fig:basic_idea_example_G}(a), when these two arrays are in form (3) of Table~\ref{tab:implementation}. Secondly, when we execute an in-memory DFS algorithm to replace $\mathcal{T}_0$ with the DFS-Tree~(i.e. $\mathcal{T}_3$ shown in Figure~\ref{fig:basic_idea_example_G}(d)) of the graph composed by $\mathcal{T}_0$ and edge batch $\{e_1,e_2,e_3\}$, the two arrays are in from (4) of Table~\ref{tab:implementation}.	
\end{example}

\begin{table}[t]
	\caption{Four forms of arryas $\mathbb{A}_1$ and $\mathbb{A}_2$, in which $T$ represents the in-memory spanning tree maintained for the input graph, $\mathcal{T}_0$ is depicted in Figure~\ref{fig:basic_idea_example_G}(a), $e_1,e_2,e_3$ are edges shown in Figure~\ref{fig:basic_idea_example_G}(a), and $\mathcal{T}_3$ is depicted in Figure~\ref{fig:basic_idea_example_G}(d).}
	\label{tab:implementation}
	\small
	\begin{center}
		\setlength{\tabcolsep}{1.4mm}{
			\begin{tabular}{c|c|c|c|c|c|c|c|c|c|c|c|c|c|c|c|c|c|c|c|c|c|c}

				\hline 
				\multirow{6}*{(1)}&\multicolumn{2}{c|}{Description}&\multicolumn{20}{l}{After initialization}\\
				\cline{2-23}
				&\multirow{2}*{$\mathbb{A}_1$}	&Index&r&a&b&c&d&f&g&h&p&q&\multicolumn{5}{c|}{Empty space Index}&\multicolumn{5}{c}{\quad}\\%\cline{2-22}
				&&Elements&-&-&-&-&-&-&-&-&-&-&\multicolumn{5}{c|}{19}&\multicolumn{5}{c}{\quad}\\
				\cline{2-23}
				&\multirow{3}*{$\mathbb{A}_2$}&Index/2&0&1&2&3&4&5&6&7&8&9&10&11&12&13&14&15&16&17&18&19\\
				&&\multirow{2}*{Elements}&-&0&1&2&3&4&5&6&7&8&9&10&11&12&13&14&15&16&17&18\\
				&&&-&-&-&-&-&-&-&-&-&-&-&-&-&-&-&-&-&-&-&-\\
				\hline \hline 
				\multirow{6}*{(2)}&\multicolumn{2}{c|}{Description}&\multicolumn{20}{l}{When $\mathcal{T}_0$ is maintained in the main memory}\\
				\cline{2-23}
				&\multirow{2}*{$\mathbb{A}_1$}	&Index&\textcolor{red}{r}&\textcolor{red}{a}&\textcolor{red}{b}&\textcolor{red}{c}&\textcolor{red}{d}&f&g&\textcolor{red}{h}&p&q&\multicolumn{5}{c|}{Empty space Index}&\multicolumn{5}{c}{\quad}\\%\cline{2-22}
				&&Elements&\textcolor{red}{17}&\textcolor{red}{16}&\textcolor{red}{14}&\textcolor{red}{13}&\textcolor{red}{12}&-&-&\textcolor{red}{11}&-&-&\multicolumn{5}{c|}{10}&\multicolumn{5}{c}{\quad}\\
				\cline{2-23}
				&\multirow{3}*{$\mathbb{A}_2$}
				&Index/2&0&1&2&3&4&5&6&7&8&9&10&\textcolor{red}{11}&\textcolor{red}{12}&\textcolor{red}{13}&\textcolor{red}{14}&\textcolor{red}{15}&\textcolor{red}{16}&\textcolor{red}{17}&\textcolor{red}{18}&\textcolor{red}{19}\\
				&&\multirow{2}*{Elements}&-&0&1&2&3&4&5&6&7&8&9&\textcolor{red}{-}&\textcolor{red}{-}&\textcolor{red}{-}&\textcolor{red}{15}&\textcolor{red}{-}&\textcolor{red}{-}&\textcolor{red}{18}&\textcolor{red}{19}&\textcolor{red}{-}\\%\cline{2-22}
				&&&-&-&-&-&-&-&-&-&-&-&-&\textcolor{red}{q}&\textcolor{red}{p}&\textcolor{red}{h}&\textcolor{red}{g}&\textcolor{red}{f}&\textcolor{red}{d}&\textcolor{red}{c}&\textcolor{red}{b}&\textcolor{red}{a}\\
				\hline
				\hline
				\multirow{6}*{(3)}&\multicolumn{2}{c|}{Description}&\multicolumn{20}{l}{When $\mathcal{T}_0$ and an edge batch $\{e_1,e_2,e_3\}$ are maintained in memory }\\
				\cline{2-23}
				&\multirow{2}*{$\mathbb{A}_1$}	&Index&r&a&\textcolor{red}{b}&c&d&f&\textcolor{red}{g}&h&\textcolor{red}{p}&q&\multicolumn{5}{c|}{Empty space Index}&\multicolumn{5}{c}{\quad}\\
				&&Elements&17&16&\textcolor{red}{9}&13&12&-&\textcolor{red}{8}&11&\textcolor{red}{10}&-&\multicolumn{5}{c|}{7}&\multicolumn{5}{c}{\quad}\\\cline{2-23}
				&\multirow{3}*{$\mathbb{A}_2$}&Index/2&0&1&2&3&4&5&6&7&\textcolor{red}{8}&\textcolor{red}{9}&\textcolor{red}{10}&11&12&13&14&15&16&17&18&19\\
				&&\multirow{2}*{Elements}&-&0&1&2&3&4&5&6&\textcolor{red}{-}&\textcolor{red}{14}&\textcolor{red}{-}&-&-&-&15&-&-&18&19&-\\
				&&&-&-&-&-&-&-&-&-&\textcolor{red}{q}&\textcolor{red}{c}&\textcolor{red}{f}&q&p&h&g&f&d&c&b&a\\
				\hline
				\hline
				\multirow{6}*{(4)}&\multicolumn{2}{c|}{Description}&\multicolumn{20}{l}{When $T$ is replaced to the DFS-Tree $\mathcal{T}_3$ of the graph composed by $\mathcal{T}_0$ and $\{e_1,e_2,e_3\}$}\\
				\cline{2-23}
				&\multirow{2}*{$\mathbb{A}_1$}	&Index&\textcolor{red}{r}&a&b&c&d&f&g&\textcolor{red}{h}&p&q&\multicolumn{5}{c|}{Empty space Index}&\multicolumn{5}{c}{\quad}\\
				&&Elements&\textcolor{red}{18}&16&9&13&12&-&8&\textcolor{red}{-}&10&-&\multicolumn{5}{c|}{17}&\multicolumn{5}{c}{\quad}\\
				\cline{2-23}
				&\multirow{3}*{$\mathbb{A}_2$} &Index/2&0&1&2&3&4&5&6&7&8&9&10&\textcolor{red}{11}&12&13&\textcolor{red}{14}&\textcolor{red}{15}&16&\textcolor{red}{17}&18&19\\
				&&\multirow{2}*{Elements}&-&0&1&2&3&4&5&6&-&14&-&\textcolor{red}{15}&-&-&\textcolor{red}{-}&\textcolor{red}{7}&-&\textcolor{red}{11}&19&-\\
				&&&-&-&-&-&-&-&-&-&q&c&f&\textcolor{red}{-}&p&h&\textcolor{red}{g}&\textcolor{red}{-}&d&\textcolor{red}{-}&b&a\\
				\hline
			\end{tabular}
		}
	\end{center}
\end{table}

One advantage, as mentioned above, of this implementation method is that it could protect EP-DFS from being affected by implementation languages, because once initialized, the above
arrays will be used until the end of the algorithms. Moreover, based on this implementation method, EP-DFS does not have to maintain an external stack when it needs to replace $T$. Instead, as (\romannumeral1) a stack is also a one-way linkedlist, and (\romannumeral2) when a node is added into the stack, there must be an edge removed from the main memory, the stack is also maintained by $\mathbb{A}_2$ as a one-way linkedlist as we manage all the empty space of $\mathbb{A}_2$. For each node $u$ in $G$, to record the most recent node $w$ visited before $u$ where $(w,u)$ belongs to $G$, we use the in-memory attribute space of the depth-first orders. That is because, we only need to preserve the order among a very small subset of $V(G)$ to update $FNN$, whose depth-first orders are in the range of $[FNN, Max(\mathcal{B}^+)]$, in order to update $FNN$ which is no more than $Max(\mathcal{B}^+)+1$. There are two ways to preserve that order. First, storing them directly on disk. Second, there is no additional disk access required when $k_2<2\times k_1$, where (\romannumeral1) $k_1$ represents the number of the nodes whose depth-first orders on $T$ are smaller than $FNN$; (\romannumeral2) $k_2$ represents the number of the nodes whose depth-first orders on $T$ are in the range of $[FNN, Max(\mathcal{B}^+)]$. The reason is discussed below. As when a node $u$ of $G$ has a depth-first order on $T$ which is no larger than $FNN$, the edges whose tails are $u$ are no need to be loaded into the main memory in EP-DFS, so that it is no need to maintain attributes $u.OF$ and $u.OD$. \textit{Of course, besides all the arrays mentioned above~($\mathbb{A}_1$ and $\mathbb{A}_2$), an integer array of length $\lceil\frac{n}{32}\rceil$ is required for recording  whether a node is marked as visited or not, when $G$ has more than $2^{31}$ nodes; otherwise, we will use the highest bit of each element in $\mathbb{A}_1$ for recording that.} 

Furthermore, there are also two benefits of our implementation method. First, procedure Rearrangement is not related to any input or output I/Os. One reason is that $T$ is stored in the main memory, where $T$ represents the input spanning tree that needs to be rearranged. Procedure Rearrangement rearranges all the out-neighborhoods of the nodes on $T$ only based on their weights, as discussed in Section~\ref{sec:algo:FNNinitializatin}. Another reason is that, no external space is used to store the weights of the nodes on $T$. In fact, we use $\mathbb{A}_2$ to maintain the node weights in procedure Rearrangement. As a node $u$ in EP-DFS and its weight $\mathcal{W}(u)$ are represented by integers, we denote $(u,\mathcal{W}(u))$ as a directed edge, and store it in $\mathbb{A}_2$ by letting $\mathcal{W}(u)$ to be the rightmost child of $u$ in $\mathbb{A}_2$. At the end of procedure Rearrangement, all the edges $(u,\mathcal{W}(u))$ could be easily and efficiently removed from $\mathbb{A}_2$. Second, it gives EP-DFS a chance to avoid the operation of sorting edge batches in procedure Indexing, without additional memory space requirement. For clarity, when the edge sorting operation required in procedure Indexing is utilized, the main memory only maintains a spanning tree $T$ for $G$, and no in-memory DFS algorithm is used in procedure Indexing. Each time it gets a set of edges $\mathcal{E}$ which is a subset of $E(G)$ and contains at most $2\times n$ edges. EP-DFS could use $\mathbb{A}_2$ to maintain all the elements of the edges of $\mathcal{E}$, by storing all the elements of $\mathbb{A}_1$ and $\mathbb{A}_2$ on disk temporarily and reinitializing these two arrays. Each time at most $2n$ edges can be loaded into the main memory since $\mathbb{A}_2$ in this implementation only has the ability of storing $2n$ edges. It is worth noting that, after $2n$ edges are loaded into the main memory, these edges could be immediately output on disk, because we already the out-neighbors of each node on the graph composed by the edges in $\mathcal{E}$.

This paper does not present the asymptotic upper bounds for the time and I/O consumption of EP-DFS. Because it is too complicated and deserves another paper, which will be our
future work. Given the size of the available
memory space, the I/O and CPU costs of EP-DFS are related
to (1) the initialized value of parameter $FNN$, in Line~\ref{line:star:FNN} of Algorithm~\ref{algo:up_down}; (2) the total iteration times of the loop in Lines~\ref{line:star:while_true}-\ref{line:star:end_while} of Algorithm~\ref{algo:up_down}; (3) the number of the edges pruned from $\mathcal{N}^+\negmedspace$\textit{-index} after executing procedure RoundI$\&$Reduction. Since the distribution of the edges and nodes in $G$ is unknown, and the order of $E(G)$ stored on disk is also unknown, it is hard to estimate the initialized value of parameter $FNN$ and the number of the edges contained in $\mathcal{N}^+\negmedspace$\textit{-index} in Line~\ref{line:star:index}. The number of the edges that contained in the index affects the total iteration times of the loop in Lines~\ref{line:star:while_true}-\ref{line:star:end_while} of Algorithm~\ref{algo:up_down}. Also, if some edges are pruned from $\mathcal{N}^+\negmedspace$\textit{-index} by procedure RoundI$\&$Reduction, then the total iteration times of such loop will also be affected. 

Our current results are shown below, which is under one assumption that the edges of $G$ are evenly distributed on disk. Supposing $B$ is the block size, $c$ is the value of parameter $FNN$ returned by Procedure~\ref{algo:intialRound}, and $p=\frac{n-c}{n}<1$. First of all, the I/O consumption of Procedure~\ref{algo:intialRound} is $O(\frac{2m}{B})=O(\frac{m}{B})$ since it has to sequentially loads $G$ into the main memory twice. The time consumption of Procedure~\ref{algo:intialRound} is $O(\lceil\frac{m}{n}\rceil\times(2n+n)+\lceil\frac{m}{n}\rceil\times(2n))=O(m)$, because (\romannumeral1) the in-memory spanning tree replacing operation scans at most $2n$ edges; (\romannumeral2) Procedure~\ref{algo:Rearrangement} requires $O(n)$ time, which is discussed in the followings; (\romannumeral3) each loop of Procedure~\ref{algo:intialRound} runs $\lceil\frac{m}{n}\rceil$ times. Assuming the nodes in $G$ are $v_1,v_2,\dots, v_n$, and in one invocation Procedure~\ref{algo:Rearrangement}, each node $v_i$~($i\in[1,n]$) of $G$ has $k_i$ out-neighbors on $T$. Thus, the time consumption of Procedure~\ref{algo:Rearrangement} is $O\big(\Sigma_{i\in[1,n]} \lceil\frac{k_i}{10^4}\rceil 10^4\log{10^4}\big)=O\big(\Sigma_{i\in [1,n]}k_i\big)=O(n)$. Secondly, Procedure~\ref{algo:Indexing} requires $O\big(\frac{m+2(p^2m)}{B}\big)=O(\frac{p^2m}{B})$ space on disk, if the edges are evenly distributed in $G$. That is because, when the edges are evenly distributed in $G$, the maximum number of the edges in $\mathcal{N}^+\negmedspace$\textit{-index} is $(\frac{n-c}{n})^2m\negmedspace=\negmedspace p^2m$. Thus, when the edges are evenly distributed in $G$, the external space required by $\mathcal{N}^+\negmedspace$\textit{-index} is at most $3\negmedspace\times\negmedspace p^2m$ bits, as discussed in Section~\ref{sec:algo:Bobtain}. The time consumption of Procedure~\ref{algo:Indexing} includes (\romannumeral1) $O(\frac{m}{B})$, the time cost of scanning all the edges in $G$, (\romannumeral2) $O(|\mathcal{E}_i|)$, the time cost of sorting the edges in $\mathcal{E}_i$~(Line~\ref{line:indexing:for_if2}, Procedure~\ref{algo:Indexing}), (\romannumeral3) $O(\frac{2|\mathcal{E}_i|}{m})$, the time cost of storing $\mathcal{E}_i$ on disk and accessing $\mathcal{E}_i$ from disk, (\romannumeral4) $O(|\mathcal{E}_0|+|\mathcal{E}_1|+\dots+|\mathcal{E}_i|)=O(p^2m)$, the time cost of merging all the ordered edge lists $\mathcal{E}_0, \mathcal{E}_1, \dots, \mathcal{E}_i$ as demonstrated in Line~\ref{line:indexing:S_v} and (\romannumeral5) $O(p^2m)$, the time cost of compressing $\mathcal{S}_v$. Hence, the time cost of Procedure~\ref{algo:Indexing} is $O(\frac{m}{B}+p^2m)$. Thirdly, the time and I/O costs of Procedure~\ref{algo:ObtainingEdges} are obvious which are $O(n)$ and $O(\frac{n}{m})$, respectively. Then, according to the above discussions, the time and I/O costs of function RoundI are $O(\lceil\frac{p^2m}{n}\rceil(n))=O(p^2m)$ and $O(\frac{p^2m}{B})$, respectively. The time and I/O costs of function RoundI$\&$Reduction are $O(p^2m)$ and $O(\frac{p^2m}{B})$, respectively.

\subsection{Correctness analysis} \label{sec:algo:correctness}

In this section, we present the correctness analysis for Algorithm~\ref{algo:naive_up_down} and Algorithm~\ref{algo:up_down}. Assuming Min$(i,j)$ represents the smaller value of $i$ and $j$.

%Firstly, Theorem~\ref{theorem:reduction_edges} states that the naive EP-DFS could find a DFS-Tree of $G$ correctly.

Theorem~\ref{theorem:min_initialize} proves the correctness of the initialization of parameter $FNN$ in Section~\ref{sec:algo:FNNinitializatin}.

\begin{theorem}\label{theorem:min_initialize}
	$\Upsilon(T_k)\geq\mathcal{C}(T,T_k)$, where (\romannumeral1) $E(G)$ is divided into a series of edge batches $B_1, B_2, \dots, B_k$, where $B_1\cup B_2\cup\dots\cup B_k=E(G)$ and if $1\leq i< j\leq k$, then $B_i\cap B_j=\phi$; (\romannumeral2) $\forall i\in [1,k]$, $T_i$ represents the DFS-Tree of the graph composed by $T_{i-1}$ and $B_i$, assuming $T$ is in the form of $T_0$.
\end{theorem}	
\begin{proof}
	According to Definition~\ref{def:S(T,T_i)} and Definition~\ref{def:U(T)}, the statement is correct, iff, the following statement is correct: ``For any edge $e$ in $G$, if the tail of $e$ has a depth-first order that is smaller than $\mathcal{C}(T,T_k)$, then $e$ is not a forward cross edge of $G$ as classified by $T_k$.'' Without loss of generality, we assume that (1) an edge $(x,y)\in B_i$; (2) $d\negthinspace f\negthinspace o(x,T_k)<\mathcal{C}(T,T_k)$. Thus, based on (1), $(x,y)$ is not a forward cross edge as classified by $T_i$, since $(x,y)$ belongs to batch $B_i$ and $T_i$ is the DFS-Tree of the graph composed by $T_{i-1}$ and $B_i$. Besides, based on (2) and Definition~\ref{def:S(T,T_i)}, $d\negthinspace f\negthinspace o(x,T)=d\negthinspace f\negthinspace o(x,T_1)=\dots=d\negthinspace f\negthinspace o(x,T_k)$, that is, the depth-first order of $x$ is unchanged during the whole process. 
	
	\noindent Hence, for all the nodes $w$, if $d\negthinspace f\negthinspace o(w,T_k)<d\negthinspace f\negthinspace o(x,T_k)$, then $d\negthinspace f\negthinspace o(w,T)=d\negthinspace f\negthinspace o(w,T_1)=\dots=d\negthinspace f\negthinspace o(w,T_k)$, because of $\mathcal{C}(T,T_k)<d\negthinspace f\negthinspace o(x,T_k)$.  Since $(x,y)$ is not a forward cross edge as classified by $T_i$, it could be a tree edge, a forward edge, a backward edge, and a backward cross edge as classified by $T_i$. For one thing, when $(x,y)$ is a tree/forward edge as classified by $T_i$, that is $d\negthinspace f\negthinspace o(x,T_i)<d\negthinspace f\negthinspace o(y,T_i)$ and $x$ is one of the ancestors of $y$ on $T_i$. To be a forward cross edge as classified by $T_k$, node $y$ should be removed from the subtree rooted at $x$, since a forward cross edge has a tail and a head which are from diffident subtrees of $T$, as discussed in Definition~\ref{def:edge_types}. That is the depth-first order of $y$ should be increased to be larger than the that of the right brother of $x$, which is impossible according to Stipulation~\ref{sti:stipulation}. For another, when $(x,y)$ is a backward~(cross) edge as classified by $T_i$. In this case $(x,y)$ should still be a backward edge as classified by any tree among $T_{i+1},T_{i+2},\dots,T_k$, since the depth-first order of $y$ is smaller than that of $x$ during the whole process as discussed above. \hfill$\Box$
\end{proof}

The correctness of Procedure~\ref{algo:Rearrangement} is proved in Theorem~\ref{theorem:node_weitht}.

\begin{theorem}\label{theorem:node_weitht}
	For a node $u$ in $T$, when we compute the weight of $u$, the weights of the children of $u$ all have been computed, in procedure Rearrangement.
\end{theorem}

\begin{proof}
	Based on the definitions in Section~\ref{sec:preliminaries}, the total depth-first order $\mathcal{O}_d$ of a tree $T$ is also the preorder~\cite{10.5555/500824} of $T$. Thus, if procedure Rearrangement computes the nodes in $T$ by the reverse order of $\mathcal{O}_d$, the statement is valid. \hfill$\Box$
\end{proof} 

Then, the correctness of Algorithm~\ref{algo:up_down} updating $FNN$ in each iteration is proved in Theorem~\ref{theorem:E+}.

\begin{theorem}\label{theorem:E+}
	$\Upsilon(T_{\mathcal{B}^+})\geq Min\big(\mathcal{C}^+(T,T_{\mathcal{B}^+}),Max(\mathcal{B}^+)+1\big)$, if $\Upsilon(T)\geq FNN$, where $T_{\mathcal{B}^+}$ is the DFS-Tree of the graph composed by $T$ and $\mathcal{B}^+$.
\end{theorem}
\begin{proof}
	Firstly, it could be easily proved that $\Upsilon(T_{\mathcal{B}^+})\geq Min\big(\mathcal{C}(T,T_{\mathcal{B}^+}),Max(\mathcal{B}^+)+1\big)$, in the same way that we prove Theorem~\ref{theorem:reduction_edges}. Secondly, assuming $d\negthinspace f\negthinspace o(u,T)\in [FNN,Max(\mathcal{B}^+)]$ and $d\negthinspace f\negthinspace o(u,T_{\mathcal{B}^+})\leq Max(\mathcal{B}^+)$. We prove the statement by contradiction. Supposing edge $(u,x)$ is a forward cross edge of $G$ as classified by $T_{\mathcal{B}^+}$. However, as $d\negthinspace f\negthinspace o(u,T)\in[FNN,Max(\mathcal{B}^+)]$, then $(u,x)\in\mathcal{B}^+$, which contradicts to our assumption according to Stipulation~\ref{sti:stipulation}. \hfill$\Box$
\end{proof}

Theorem~\ref{theorem:roundI} states the correctness of our optimization algorithm.

\begin{theorem} \label{theorem:roundI}
	If procedure RoundI or procedure RoundI$\&$Reduction returns $T$ and $FNN$, then $\Upsilon(T)\geq FNN$.
\end{theorem}
\begin{proof}
	It can be derived from Theorem~\ref{theorem:E+} that: ``In the $i$th iteration of EP-DFS, $\forall e=(u,v)\in E(G)$, if $d\negthinspace f\negthinspace o(u,T)< FNN$ and $d\negthinspace f\negthinspace o(v,T)\leq FNN$, then $e$ cannot be a forward cross edge as classified by $T$ in the $j$th~($j\geq i$) iteration of $G$ in EP-DFS''. Thus, the correctness of this statement could be proved based on Theorem~\ref{theorem:min_initialize}, which is discussed below. For one thing, our $\mathcal{N}^+\negthickspace$\textit{-index} is obtained after the value of $FNN$ is given. Moreover, in the computation process of $\mathcal{N}^+\negthickspace$\textit{-index}, $\forall e=(u,v)\in E(G)$, $e$ is contained in $\mathcal{N}^+\negthickspace$\textit{-index}, iff, $d\negthinspace f\negthinspace o(u,T)\geq FNN$ and $d\negthinspace f\negthinspace o(v,T)>FNN$. For another, in procedure RoundI or procedure RoundI$\&$Reduction, an edge $e=(u,v)$ is discarded, iff, $d\negthinspace f\negthinspace o(u,T)< FNN$ and $d\negthinspace f\negthinspace o(v,T)\leq FNN$. \hfill$\Box$
\end{proof}

The termination proof and the correctness proof of EP-DFS are given in Theorem~\ref{theorem:terminated_correctness}.

\setcounter{theorem}{5}
\begin{theorem}\label{theorem:terminated_correctness}
	Algorithm~\ref{algo:up_down} can be finally terminated, and returns $T$ as a DFS-Tree of $G$.
\end{theorem}
\begin{proof}
	\textit{Termination proof.} The loop in Lines~\ref{line:star:while_true}-\ref{line:star:end_while} of Algorithm~\ref{algo:up_down} ends when $FNN\geq n$. In other words, Algorithm~\ref{algo:up_down} can be terminated iff the value of $FNN$ could exceed $n-1$, which can be proved based on the following two points. One is, based on the discussions on Section~\ref{sec:preliminaries}, there exists an edge batch $\mathcal{B}^+$, where $|\mathcal{B}^+|<n$. Another is, according to Theorem~\ref{theorem:E+}, the value of $FNN$ must increase in each iteration of the loop~(Lines~\ref{line:star:while_true}-\ref{line:star:end_while} of Algorithm~\ref{algo:up_down}).  
		
	\noindent\textit{Correctness proof.} The correctness of EP-DFS is equivalent to ``In each iteration of Algorithm~\ref{algo:up_down}, if $\forall (u,v)\in E(G)$, if $d\negthinspace f\negthinspace o(u,T)<FNN$ or $d\negthinspace f\negthinspace o(v,T)\leq FNN$, $(u,v)$ is not a forward cross edge as classified by $T$''. According to Definition~\ref{def:S+(T, T)}, Theorem~\ref{theorem:E+}, and Theorem~\ref{theorem:roundI} then the above statement is certainly correct. \hfill$\Box$
\end{proof}

\section{Experimental Evaluation}\label{sec:experiments}
In this section, we evaluate the performance of the proposed algorithm, EP-DFS, against the EB-DFS and DC-DFS algorithms, on both synthetic and real graphs. Specifically, we are interested in the efficiency and the number of I/Os for each algorithm on each graph, where we measure the former by the running time, and the latter by the total size of disk accesses. Besides, with the assumption that each input graph are stored on disk in the form of edge list, we are also interested on the effects of the different disk edge storage methods, i.e. random list~(\textit{default storage method}) and adjacency list~(the directed edges with the same tail are continuous stored in disk.). Our experiments run on a machine with the intel i7-9700 CPU, 64 GB RAM and 1TB disk space. All the algorithms in our experiments are implemented by Java. Note that, we limit each experiment within eight hours, and we restrict that at most $2n$ edges could be hold in the main memory, as discussed in our problem statement~(Section~\ref{sec:preliminaries}).

\textit{Datasets.} We utilize various large-scale datasets including both real and synthetic graphs. The storage method for all utilized graphs is default to random list on disk.

The real datasets are presented in Table~\ref{tab:real_datasets}, which include one relatively small graph, two social networks, and several large crawls or massive networks from different domains\footnote{https://github.com/google/guava/wiki/InternetDomainNameExplained}. \textit{cnr-2000} is a relatively small crawl based on the Italian CNR domain. \textit{amazon-2008} describes the similarity among the books of Amazon store, which is a symmetric graph. \textit{hollywood-2011} is one of the most popular social graphs, in which the nodes are actors, and each edge links two actors appeared in a movie together. \textit{eu-2015-host} is the host~(the maximum number of pages per host is set to 10M) graph of eu-2015, which is a large snapshot of the domains of European Union countries, taken in 2015 by BUbiNG~\cite{BMSB} and starting from the site ``http://europa.eu/''. \textit{uk-2002} is a 2002 crawl of the .uk domain performed by UbiCrawler~\cite{BCSU3}. \textit{gsh-2015-tpd} is the graph of top private domains of gsh-2015, which is a large snapshot of the web taken in 2015 by BUbiNG, similar to graph eu-2015 but without any domain restriction. \textit{it-2004} is a fairly large crawl of the .it domain. \textit{sk-2005} is 2005 crawl of the .sk domain performed by UbiCrawler. All the above utilized datasets can be accessed from the website ``http://law.di.unimi.it/datasets.php''. % each of which  we decompress and restored in disk in the form of random edge list.

\begin{table*}[t]
	
	\caption{The experimental results on the real datasets, where (\romannumeral1) LCC is an abbreviation for largest connected component; (\romannumeral2) the running time~(RT) is in seconds; (\romannumeral3) the number of I/Os and the index size~(IS) are in megabytes; (\romannumeral4) ``-'' indicates that the method timed out on this dataset.}
	\label{tab:real_datasets}
	\small
	\begin{center}
		\setlength{\tabcolsep}{1.5mm}{
			\begin{tabular}{c|cccc|cc|cc|cccc}
				\hline \multirow{2}*{Dataset}& \multirow{2}*{$n/10^6$} &\multirow{2}*{$m/10^6$}&\multirow{2}*{$m/n$}&\multirow{2}*{LCC$/10^6$} & \multicolumn{2}{c|}{EB-DFS} &\multicolumn{2}{c|}{DC-DFS }& \multicolumn{3}{c}{EP-DFS}\\
				&	& & & & RT & I/O& RT & I/O& RT & I/O & IS\\
				\hline
				cnr-2000	&$0.33$&$3$	&$9.88$	&$0.11 (34.4\%)$&$190$&$10,529$	&$47$&$1,809$	&$7$&$225$			&$1.6$\\	
				amazon-2008	&$0.74$&$5$	&$7.02$	&$0.63 (85.3\%)$&$125$&$4,445$		&$59$&$1,434$	&$18$&$318$		&$1.5$\\
				hollywood-2011&$2.18$&$229$&$105$	&$1.92(87.9\%)$&-&-				&$3,539$&$118,208$&$277$&$10,318$		&$4.4$\\
				eu-2015-host&$11.3$&$387$	&$34.4$	&$6.51(57.8\%)$&$15,423$&$908,045$&$25,218$&$1,346,455$&$759$&$29,087$	&$146$\\
				uk-2002		&$18.5$&$298$	&$16.1$	&$12.1(65.3\%)$&$19,821$&$923,335$	&$22,934$&$456,917$&$833$&$32,837$	&$131$\\
				gsh-2015-tpd&$30.8$&$602$	&$19.5$	&$20.0(64.9\%)$&-&-				&-&-			&$1,466$&$39,833$		&$102$\\
				it-2004		&$41.3$&$1,151$&$27.9$	&$29.9(72.3\%)$&-&-				&-&-			&$3,927$&$163,805$	&$500$\\
				sk-2005		&$50.6$&$1,949$&$38.5$	&$35.9(70.9\%)$&-&-				&-&-			&$5,723$&$222,836$	&$843$\\
				\hline				
			\end{tabular}
		}
	\end{center}
\end{table*}

The synthetic datasets are randomly generated~\cite{DBLP:journals/csur/DrobyshevskiyT20}, according to Erdös-Rényi~(ER) model~(\textit{default model}) and scale-free~(SF) model. Firstly, for a dataset $G=(V,E)$ in ER model, we randomly and repeatedly generate an edge $e=(u,v)$ that $u,v\in V$ and $u\neq v$, where the edges in $E$ are unique. Then, for the datasets following SF model, the generation method is in the way in~\cite{SF}, where the parameters $p$, $q$ and $m$ are set to $0.9$, $0$ and $1$, respectively. 

\textit{Comparison algorithms and implementation details.} In literature, many algorithms are proposed for addressing the DFS problem~\cite{DBLP:books/daglib/0037819}. However, only a few of them could be used on semi-external memory model, since it is non-trivial to solve the DFS problem under the restriction that only a spanning tree of the input graph could be maintained in memory. These algorithms include EE-DFS, EB-DFS, and DC-DFS, as discussed in Section~\ref{sec:related_works}. Among these algorithms, EE-DFS is extremely inefficient, because it processes the edges of $G$ one by one instead of edge batches. Thus, in this section, we evaluate our EP-DFS against EB-DFS and DC-DFS.

We prefer the most efficient data structures in the limited main memory space~(as discussed in Section~\ref{sec:preliminaries} and Section~\ref{sec:algorithm}). For example, in EB-DFS, we utilize our rearrangement algorithm. The division technique used in the evaluated DC-DFS algorithm is Divide-TD, as the other division technique is inefficient reported in \cite{DBLP:conf/sigmod/ZhangYQS15}. In addition, we develop the DC-DFS algorithm based on Tarjan algorithm~(fast strongly connected component algorithm)~\cite{DBLP:journals/siamcomp/Tarjan72} and Farach-Colton and Bender Algorithm~(fast LCA algorithm)~\cite{DBLP:journals/jal/BenderFPSS05} to ensure the efficiency of DC-DFS algorithm. Note that, in EP-DFS, if $F_2-FNN<\,$Min$(100,\frac{n}{1000})$ in Line~\ref{line:star:while_if} and Line~\ref{line:star:while_else_if}, Algorithm~\ref{algo:up_down}, we say $\frac{F_2-FNN}{n}\rightarrow 0$, and we set the threshold $\gamma$ to $10\%$ by default.

\subsection{Exp 1: Performance on real large graphs} \label{sec:experiments:exp1}
%\textbf{Exp 1: Performance on real large graphs.} 
We evaluate the semi-external DFS algorithms on eight real graphs. The disk storage method for all the utilized graphs is default to random list. The evaluation results on the real datasets are demonstrated in Table~\ref{tab:real_datasets}, where we also present (external) space~(index size, IS) cost for function Indexing~(Line~\ref{line:star:index}) of EP-DFS. The indexing or restructuring time of $\mathcal{N}^+\negthickspace$\textit{-index} in EP-DFS is included in the running time~(RT) of EP-DFS. 

Table~\ref{tab:real_datasets} shows that, compared to EB-DFS and DC-DFS, our EP-DFS could achieve a great performance on the complex real datasets. In other words, EP-DFS is an order of magnitude faster than EB-DFS and DC-DFS, on the reported results. Besides, EP-DFS's I/O consumption is also greatly lower than one-tenth of EB-DFS's and DC-DFS's I/O consumption. For example, to obtaining a DFS-Tree of dataset eu-2015-host, EB-DFS costs $15,423$s and $908,045$MB I/Os; DC-DFS requires $25,128$s and $1,346,455$MB I/Os; EP-DFS could be accomplished within $759$s and $29,087$MB I/Os. That is EB-DFS and DC-DFS consume $20$ and $33$ times as much time as EP-DFS, respectively, and they require $31$ and $46$ times as much space as EP-DFS, respectively. Note that, the experimental results also reflect that the indexing process of $\mathcal{N}^+\negthickspace$\textit{-index} is efficient, and the external space that such index requires is considerably small compared to $m$ or $n$.

In addition, we test the three algorithms on hollywood-2011, uk-2002 and sk-2005, by randomly selecting edges from such datasets, as demonstrated in Figures~\ref{fig:hollywood}-\ref{fig:sk}. We vary the number of edges from $20\%m$ to $100\%m$, as shown in the x-axes of Figures~\ref{fig:hollywood}-\ref{fig:sk}. The chosen reason of such datasets against the others is as follows: (\romannumeral1) hollywood-2011 is the dataset with the highest average node degree~($\frac{m}{n}$), i.e. $105$; (\romannumeral2) uk-2002 is a relatively large graph among the given eight real graphs whose largest connected component contains more than half of nodes; (\romannumeral3) sk-2005 is the graph with the largest scale including a large connected component with about $35.9$M nodes. Specially, in order to generate a random graph $G_p=(V_p,E_p)$ of an input graph $G=(V,E)$ that $\frac{|E_p|}{m}\negthickspace=\negthickspace p$, we scan all the edges in $E$, where, for each edge $e\in E$, $e$ is selected independently, and added into $G_p$ with $p$ probability. Since the number of the edges in $G$ is huge, the size of the generated edge set $E_p$ is $p\times m$, according to the law of large numbers~\cite{Statistic}.

\begin{figure*}
	\begin{minipage}[t]{0.49\linewidth}
		
		\centering
		\renewcommand{\thesubfigure}{}
		\subfigure[(a) Efficiency]{
			\includegraphics[scale = 0.6]{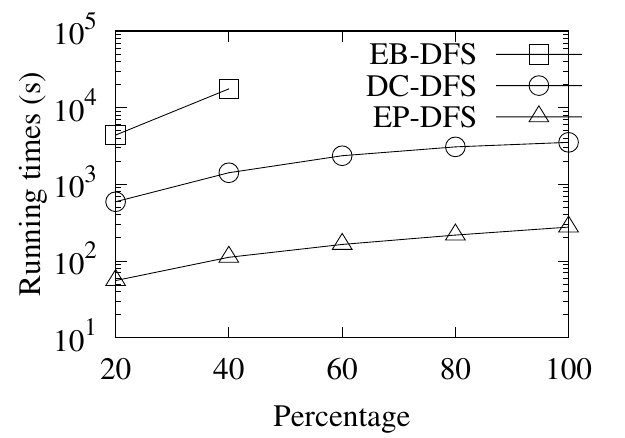}}
		\subfigure[(b) I/O]{
			\includegraphics[scale = 0.6]{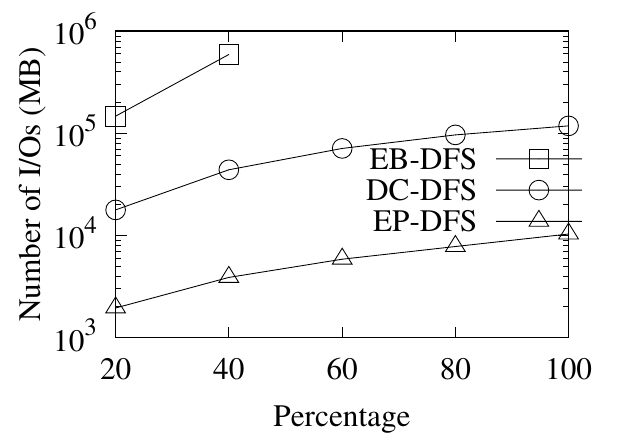}}
		\caption{The experimental results on hollywood-2011.}
		\label{fig:hollywood}
		
	\end{minipage}
	\begin{minipage}[t]{0.48\linewidth}
		
		\centering
		\renewcommand{\thesubfigure}{}
		\subfigure[(a) Efficiency]{
			\includegraphics[scale = 0.6]{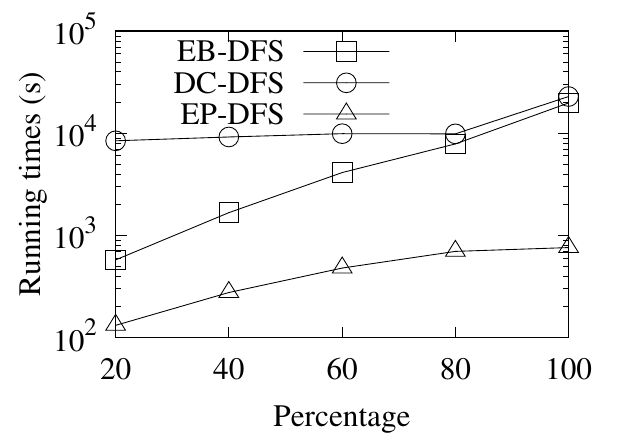}}
		\subfigure[(b) I/O]{
			\includegraphics[scale = 0.6]{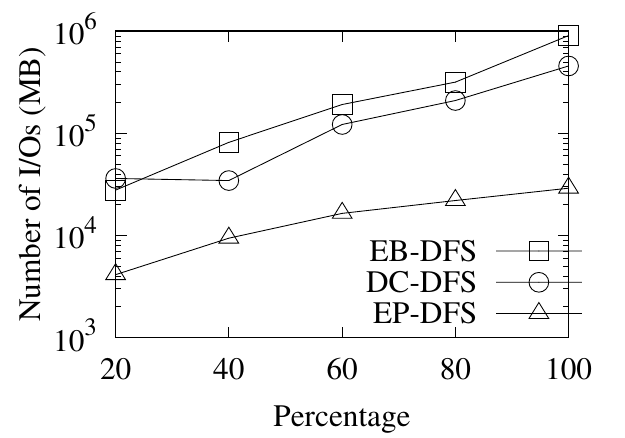}}
		\caption{The experimental results on uk-2002.}
		\label{fig:uk}
		
	\end{minipage}
\end{figure*}

\begin{figure*}
	\begin{minipage}[t]{0.49\linewidth}
		
		\centering
		\renewcommand{\thesubfigure}{}
		\subfigure[(a) Efficiency]{
			\includegraphics[scale = 0.6]{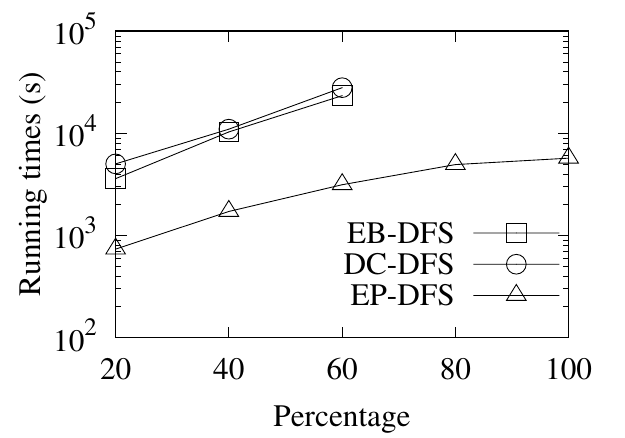}}
		\subfigure[(b) I/O]{
			\includegraphics[scale = 0.6]{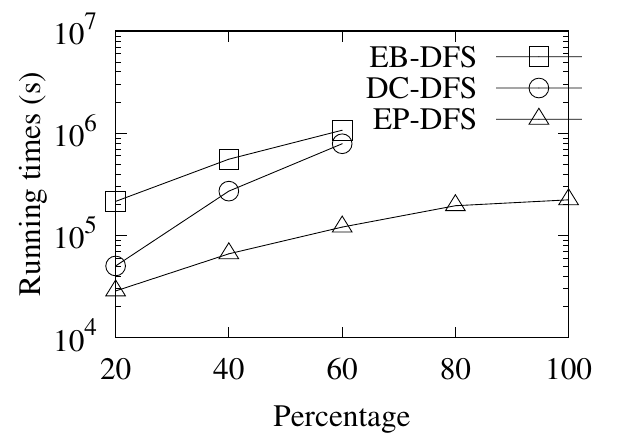}}
		\caption{The experimental results on sk-2005.}
		\label{fig:sk}
		
	\end{minipage}
	\begin{minipage}[t]{0.48\linewidth}
		
		\centering
		\renewcommand{\thesubfigure}{}
		\subfigure[(a) Efficiency]{
			\includegraphics[scale = 0.6]{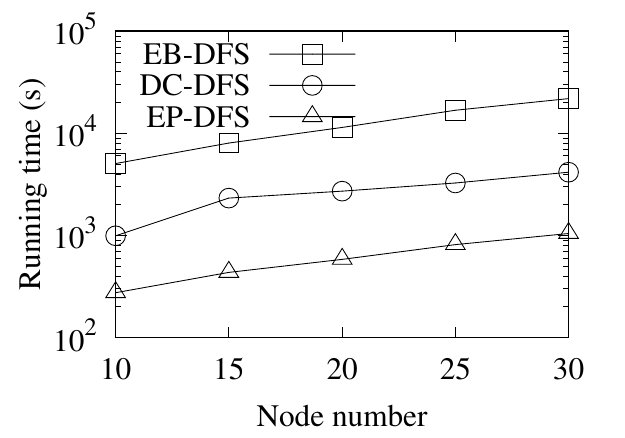}}
		\subfigure[(b) I/O]{
			\includegraphics[scale = 0.6]{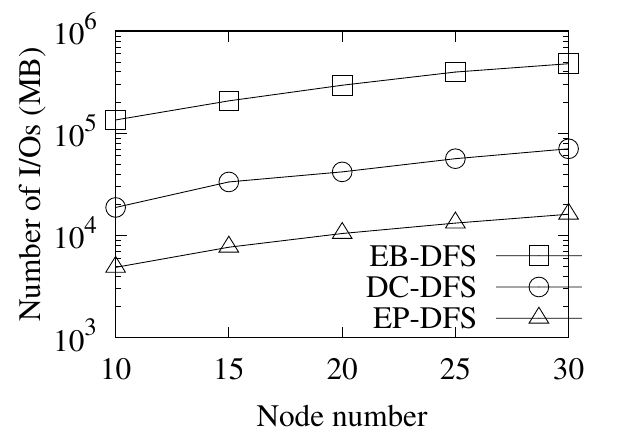}}
		\caption{Varying $n$ on the synthetic graphs of ER model.}
		\label{fig:ER_V}
		
	\end{minipage}
\end{figure*}

The experimental results in Figures~\ref{fig:hollywood}-\ref{fig:sk} confirm that our EP-DFS outperforms the traditional algorithms on the real large graphs with different structures. Firstly, in Figure~\ref{fig:hollywood}, the EB-DFS algorithm cannot construct the DFS-Tree when the generation percent $p$ exceeds $40\%$, while, even on entire hollywood-2011 dataset, the cost of EP-DFS is only about $10^2$s. The reason is that EB-DFS needs to execute function Round many times, when the structure of the input graph goes more complex, according to the discussion about the ``chain reaction'' in Section~\ref{sec:overview}. Secondly, in Figure~\ref{fig:uk}, the performance of DC-DFS is poor, which consumes more than $10^4$s for each generated graphs of uk-2002, compared to the performance of EP-DFS, which requires less than $10^3$s on the entire uk-2002 dataset. Besides, even though the I/O costs of DC-BFS on the $20\%$, $40\%$ uk-2002 graphs are less than that on the $60\%$ uk-2002 graph, the time costs are nearly the same. That is because the processes of DC-DFS on such datasets are related to random disk I/O accesses, which is discussed in Section~\ref{sec:preliminaries}. Then, in Figure~\ref{fig:sk}, both DC-DFS and EB-DFS are terminated because of the time limitation, when $p$ exceeds $60\%$ on the sk-2005 dataset.

\subsection{Exp 2: The impact of varying $n$ on synthetic graphs} \label{sec:experiments:exp2}
We vary the number of the nodes from $10,000,000$ to $30,000,000$, for the graphs in ER model, and we set the average node degree~($\frac{m}{n}$) to $10$, for each generated graph. All the graphs are stored on disk in the form of random list. The experimental results about the time and I/O consumption of the evaluated algorithms are presented in Figure~\ref{fig:ER_V}(a) and Figure~\ref{fig:ER_V}(b), respectively. As the number of nodes grows, the running time and the number of I/Os required by each evaluated algorithm increase. However, among all the algorithms, EP-DFS has the lowest increasing rate, and EB-DFS has the highest increasing rate. The reason is that, when the number of the nodes increases, the size of the entire input graph grows, i.e. from $100$M to $300$M. Since the ``chain reaction'' exists, restructuring the in-memory spanning tree to a DFS-Tree goes harder, where the invocation times of both function Round and function Reduction-Rearrangement increase in EB-DFS. In contrast, our EP-DFS, after constructing $\mathcal{N}^+\negthickspace$\textit{-index}, could avoid scanning the entire input graphs. Besides, our EP-DFS greatly reduce the number of the I/Os, which only requires about $10^4$MB total size of disk accesses.

\subsection{Exp 3: The impact of varying $\frac{m}{n}$ on synthetic graphs}\label{sec:experiments:exp3}

We vary the average degree of the nodes from $10$ to $30$, for the graphs in ER model, in this part. The node number is default to $10,000,000$, and the storage method is random list by default. The experimental results about the time and I/O cost are depicted in Figure~\ref{fig:ER_D}(a) and Figure~\ref{fig:ER_D}(b), respectively, which demonstrate that: with the increase of $\frac{m}{n}$~(the average node degree), the numbers of the running time and the disk I/O accesses are increased. Since the chain reaction exists, the number of I/Os required by algorithm EB-DFS is far beyond $10^5$, and the running time reaches the time limit, i.e. 8 hours, when $\frac{m}{n}=30$. Plus, the performance of DC-DFS is acceptable, even though that of DC-DFS is worse than that of EP-DFS which requires less than $10^3$s and $10^4$ I/Os.

\subsection{Exp 4: The impact of different disk storage methods} \label{sec:experiments:exp4} For the two kinds of disk storage algorithms, we evaluate all the algorithms on the graphs, where the node numbers are set to $10,000,000$, and we vary the average node degree from $10$ to $30$. All the graphs are synthetic datasets and in the form of ER model. In other words, we restore all the utilized synthetic datasets in Exp~3 in the form of adjacency list. The experimental results on the graphs stored in the form of random list are presented in Figure~\ref{fig:ER_D}, while that in the form of adjacency list are depicted in Figure~\ref{fig:ER_AD}. Such results show that there is an increase trend when the number of $\frac{m}{n}$ increases, no matter what disk storage method is. However, the performances of both the two traditional algorithms are worse when the input graphs are stored in the form of adjacency list, while, the performance of EP-DFS is slightly better. Furthermore, both EB-DFS and DC-DFS reach the time limit in the experiments depicted in Figure~\ref{fig:ER_AD}. Such experimental results indicate that (\romannumeral1) the different graph storage method changes the orders of the edges on disk, which affects the process of restructuring an in-memory spanning tree to a DFS-Tree; (\romannumeral2) compared to the traditional algorithms, our EP-DFS algorithm is more adaptable to the different disk-resident graph storage methods.

\begin{figure*}
	\begin{minipage}[t]{0.49\linewidth}
		
		\centering
		\renewcommand{\thesubfigure}{}
		\subfigure[(a) Efficiency]{
			\includegraphics[scale = 0.6]{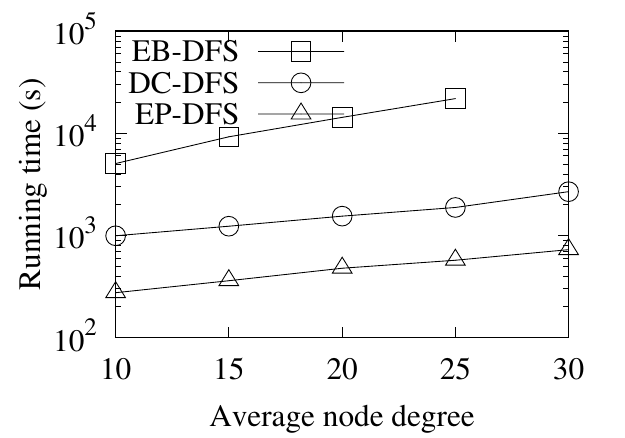}}
		\subfigure[(b) I/O]{
			\includegraphics[scale = 0.6]{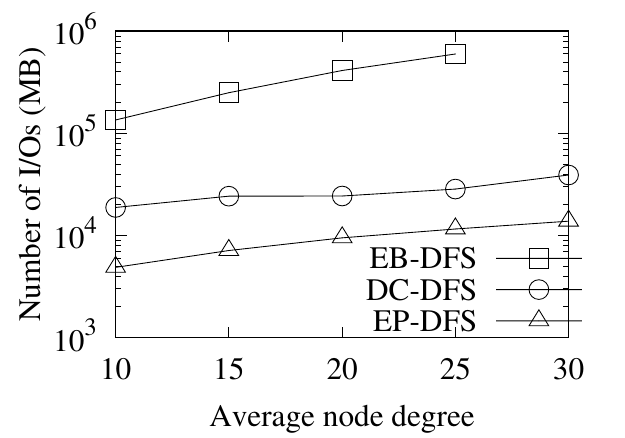}}
		\caption{Varying $\frac{m}{n}$ on the synthetic graphs of ER model.}
		\label{fig:ER_D}
		
	\end{minipage}
	\begin{minipage}[t]{0.48\linewidth}
		
		\centering
		\renewcommand{\thesubfigure}{}
		\subfigure[(a) Efficiency]{
			\includegraphics[scale = 0.6]{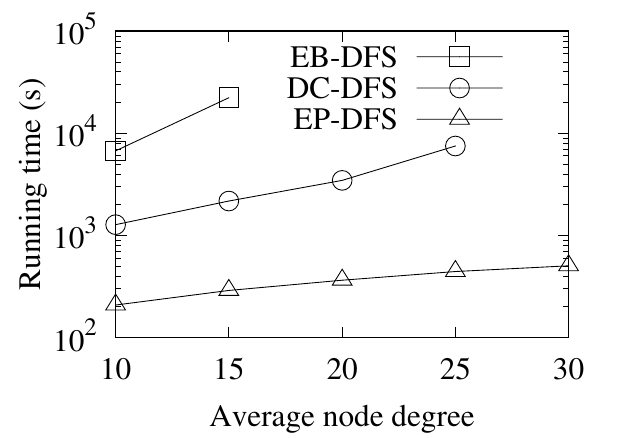}}
		\subfigure[(b) I/O]{
			\includegraphics[scale = 0.6]{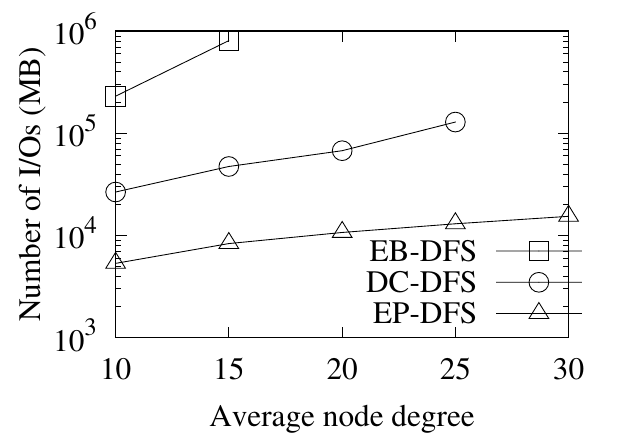}}
		\caption{The results on the ER graphs stored by adjacency list.}
		\label{fig:ER_AD}
		
	\end{minipage}
\end{figure*}

\begin{figure*}
	\begin{minipage}[t]{0.49\linewidth}
		
		\centering
		\renewcommand{\thesubfigure}{}
		\subfigure[(a) Efficiency]{
			\includegraphics[scale = 0.6]{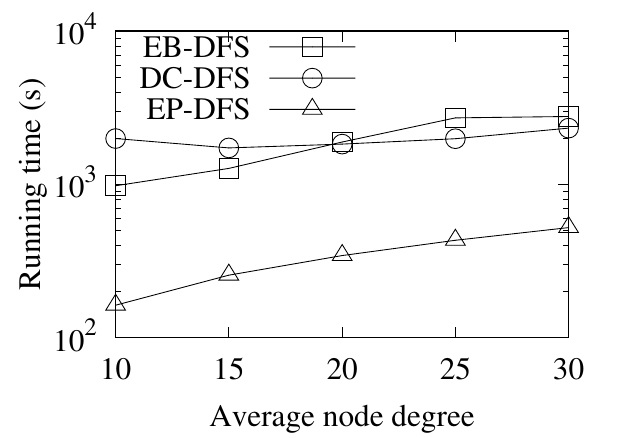}}
		\subfigure[(b) I/O]{
			\includegraphics[scale = 0.6]{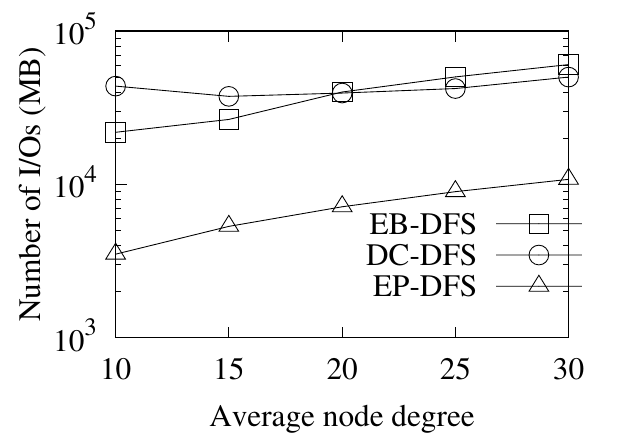}}
		\caption{Varying $\frac{m}{n}$ on the synthetic graphs of SF model.}
		\label{fig:SR_D}
		
	\end{minipage}
	\begin{minipage}[t]{0.48\linewidth}
		
		\centering
		\renewcommand{\thesubfigure}{}
		\subfigure[(a) Efficiency]{
			\includegraphics[scale = 0.6]{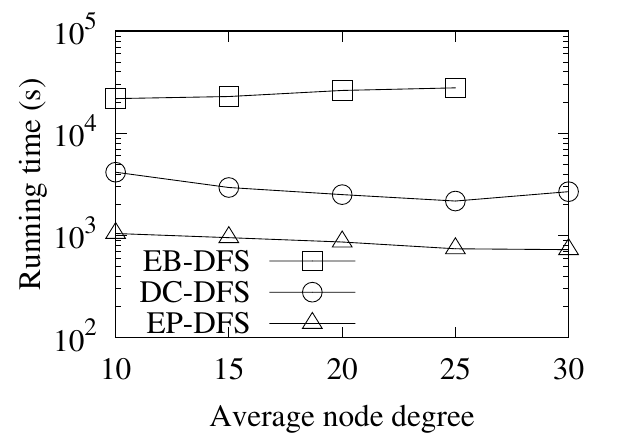}}
		\subfigure[(b) I/O]{
			\includegraphics[scale = 0.6]{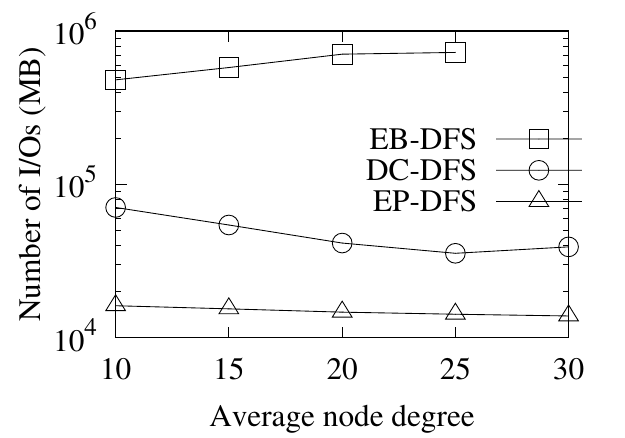}}
		\caption{Varying $\frac{m}{n}$ on synthetic graphs of ER model with fixed $m$.}
		\label{fig:ER_E}
		
	\end{minipage}
\end{figure*}

\subsection{Exp 5: The impact of different graph structures} \label{sec:experiments:exp5} 
In this part, we evaluate all the three semi-external DFS algorithms on both ER and SF graphs. Each group of the experiments runs on the graphs with fixed node number, i.e. $n=10,000,000$, in which we vary the average node degree $\frac{m}{n}$ in the range of $[10,30]$. Plus, the disk storage method of each graph is default to random list. The experimental results of the running time and the number of disk I/O accesses on ER graphs are depicted in Figure~\ref{fig:ER_D}(a) and Figure~\ref{fig:ER_D}(b), respectively, and that on SF graphs are demonstrated in Figure~\ref{fig:SR_D}(a) and Figure~\ref{fig:SR_D}(b), respectively. The performances of the evaluated algorithms on the SF graphs are better than that on the ER graphs. However, the time and I/O consumption of DC-DFS on certain SF graphs are higher than that of the other two algorithms. The reason is that, the division process of DC-DFS is hard on the SF graphs generated in the way of~\cite{SF}, in which, the more links that a node $v$ is connected to, the higher the probability that it adds a new edge related to $v$.

\subsection{Exp 6: The impact of fixing $m$ on synthetic graphs} \label{sec:experiments:exp6}
 Since the performance of the semi-external DFS algorithms is related to the scales of the input graphs, we are interested in the performance of the three algorithms on the synthetic graphs with fixed edge number. Specifically, in this part, we set $m$ to $300,000,000$ for each generated graph, and vary average degree $\frac{m}{n}$ from $10$ to $30$, as depicted in Figure~\ref{fig:ER_E}. The disk storage method is default to random list. The node numbers of the graphs are $30$M, $20$M, $15$M, $12$M and $10$M, respectively. The experimental results of the running time and the number of required I/Os are demonstrated in Figure~\ref{fig:ER_E}(a) and Figure~\ref{fig:ER_E}(b), respectively. According to the depicted results, the performance of EB-DFS algorithm goes worse, when the average degree of input graph increases. Especially, when $\frac{m}{n}= 30$ and $m= 30$, the EB-DFS reaches the time limit. Because, with the fixed size of $m$, the larger number of $\frac{m}{n}$, the more complex the graph structure is, which causes numerous chain reactions in the restructuring process. In contrast, EP-DFS and DC-DFS could address the given input graphs with higher efficiency and less I/Os, according to Figure~\ref{fig:ER_E}, with the increase of the average degree.

\section{Conclusion} \label{sec:conclusion}
This paper is a comprehensive study of the DFS problem on semi-external environment, where the entire graph cannot be hold in the main memory. This problem is widely utilized in many applications. Assuming that at least a spanning tree $T$ can be hold in the main memory, semi-external DFS algorithms restructure $T$ into a DFS-Tree of $G$ gradually. This paper discusses the main challenge of the non-trivial restructuring process with theoretical analysis, i.e. the ``chain reaction'', which causes the traditional algorithms to be inefficient. Then, based on the discussion, we devise a novel semi-external DFS algorithm, named EP-DFS, with a lightweight index $\mathcal{N}^+\negthickspace$\textit{-index}. The experimental evaluation on both synthetic and real large datasets confirms that our EP-DFS algorithm significantly outperforms traditional  algorithms. Our future work is to present the asymptotic upper bounds of the time and I/O costs of EP-DFS. It is interesting but intricate, since the performance of EP-DFS is affected by many interrelated factors as demonstrated in this paper.

\section*{Acknowledgments}
This paper was partially supported by NSFC grant 61602129.

%\ifCLASSOPTIONcompsoc
%% The Computer Society usually uses the plural form
%\section*{Acknowledgments}
%This paper was partially supported by NSFC grant U1509216, U1866602, 61602129 and Microsoft Research Asia. 
%\else
%% regular IEEE prefers the singular form
%\section*{Acknowledgment}
%\fi
%
%% Can use something like this to put references on a page
%% by themselves when using endfloat and the captionsoff option.
%\ifCLASSOPTIONcaptionsoff
%\newpage
%\fi
%
%\normalem

%\bibliographystyle{plain}
%\bibliography{IEEEexample}





%\begin{thebibliography}{00}


%\section*{References}

%\bibliographystyle{elsarticle-num}
\bibliographystyle{abbrv}
%\bibliography{IEEEexample}

\end{document}